\documentclass[fullpage]{article}

\usepackage{color}
\usepackage{bm}
\usepackage{bbm}
\usepackage{nicefrac}
\usepackage{enumerate}
\usepackage{amsmath,amssymb}
\usepackage{amsthm}
\usepackage{authblk}
\usepackage{graphicx}

\renewcommand{\v}[1]{\bm{#1}}
\newcommand{\E}{\mathbb{E}}
\newcommand{\eq}[1]{equation (\ref{eq:#1})}

\newcommand{\noeq}[1]{(\ref{eq:#1})}
\DeclareMathOperator{\diag}{diag}
\newcommand{\norm}[1]{\left\lVert#1\right\rVert_{1}}
\renewcommand{\d}[1]{\,d#1}
\newtheorem{assumption}{Assumption}
\newtheorem{theorem}{Theorem}
\newtheorem{corollary}{Corollary}
\newtheorem{definition}{Definition}
\newtheorem{proposition}{Proposition}
\newtheorem{hypothesis}{Hypothesis}

\usepackage{natbib}
 \bibpunct[, ]{(}{)}{,}{a}{}{,}%
\begin{document}

\title{Systemic Greeks: Measuring risk in financial networks}
\author[1,2]{Nils Bertschinger \thanks{bertschinger@fias.uni-frankfurt.de}}
\author[1]{Julian Stobbe \thanks{jstobbe@fias.uni-frankfurt.de}}
\affil[1]{Department of Systemic Risk, Frankfurt Institute for Advanced Studies, Frankfurt am Main, Germany}
\affil[2]{Department of Computer Science, Goethe University, Frankfurt am Main, Germany}

\maketitle

\begin{abstract}
  Since the latest financial crisis, the idea of systemic risk has
  received considerable interest. In particular, contagion effects
  arising from cross-holdings between interconnected financial firms
  have been studied extensively. Drawing inspiration from the field of
  complex networks, these attempts are largely unaware of models and
  theories for credit risk of individual firms. Here, we note that
  recent network valuation models extend the seminal structural risk
  model of \cite{Merton1974}. Furthermore, we formally compute
  sensitivities to various risk factors -- commonly known as Greeks --
  in a network context. In particular, we propose the network $\Delta$
  as a quantitative measure of systemic risk and illustrate our
  findings on some numerical examples.

\end{abstract}

\section{Introduction}

Since the latest financial crisis, the idea of systemic risk has
received considerable interest. Systemic risk denotes the risk
that large parts of the entire financial system, e.g. important
markets for credit or liquidity, collapse. Generally, two approaches
can be distinguished. On one hand, quantitative measures of
systemic risk have been proposed
\citep{Acharya2010,Adrian2016,Brownlees2016} and estimated on market
data. These measures strive to capture the statistical phenomena
associated with large market disruptions, e.g. tail-correlations or
conditional shortfalls. As such, they are estimated from market data
without considering fundamentals about firm's portfolio
compositions. On the other hand, complex network models start from the
structure of cross-holdings. Based on simple yet plausible
assumptions about insolvency resolution the impact of individual
defaults on the entire financial system is studied
\citep{Eisenberg2001,Gai2010,Upper2011,Battiston2012,Barucca2016}. As
such, these models are mainly investigated by analytic and numeric
means on random networks loosely resembling real financial
structures. Due to their stylized nature they are mostly of
illustrative nature, yet providing important insights into the
dynamics of financial contagion, e.g. on the appearance of a contagion
window and the robust-yet-fragile nature of financial systems.

\subsection{Network valuations}

The above network models analyze and simulate the contagion arising
from defaults/insolvencies of direct counterparties. To this end, they
focus on the cash flows of repaid liabilities -- in full or in part --
when contracts are settled. This idea is most vividly expressed in the
work by \cite{Eisenberg2001} as the ``clearing payment vector''. An
alternative viewpoint underlies the work by \cite{Suzuki2002} who
considers a network of firms with cross-holdings of debt and
equity. Repayment of liabilities, i.e. debt cross-holdings,
cannot be considered from a cash flow perspective alone, as the value
of equity cross-holdings needs to be known as well. The model
therefore extends the seminal \cite{Merton1974} model which expresses
debt and equity values as derivative contracts on a single firms
assets to multiple firms with cross-holdings. From this perspective
instead of a consistent set of clearing payments, all contracts have
to be valued consistently. Interestingly, this strand of research has
developed mostly unaware of clearing models, even re-discovering
Suzuki's results several times \citep{Elsinger2009,Fischer2014}. Only
recently \cite{Barucca2016} have shown that the idea of network
valuation actually unifies many models, including the ones by
\cite{Eisenberg2001} and \cite{Battiston2012}, under a common
framework. Here, we build on this framework to investigate risk in
financial networks.

\subsection{Greeks}

In the context of a single firm, \cite{Merton1974} has shown that
equity can be considered as a long call option and debt insurance as a
short put option on the firm's asset value. Thus connecting credit
risk with option pricing. Such {\em structural credit risk} models
have since been developed and used extensively to assess and
manage credit risk. Especially Moody's KMV model has become an industry
standard in this respect and is routinely used to recover implied
probabilities of default from market price data.

Furthermore, option pricing provides a wide range of risk measures and
management techniques. The Greeks
\citep{Haug2003a,Haug2003b} are widely used to quantify sensitivities
to several risk factors, e.g. interest rate, volatility and many
more. Consequently, they are routinely used in practice to evaluate
and hedge the risk of option portfolios. Delta hedging provides a
well-known technique to reduce portfolio risk and builds on the
option's $\Delta$, measuring the sensitivity of the option price to
changes in the spot price of the underlying. Its importance is
reflected in the fact that options are routinely quoted in terms of
their $\Delta$. As the other Greeks, $\Delta$ is formally defined as a
partial derivative.

In the context of network valuations, partial derivatives have been
used a few times to quantify sensitivities, be it to asset values
\citep{Liu2010,Demange2018} or changes in the value of cross-holdings
\citep{Feinstein2017,Karl2015}. Here, we extend this work in several
ways. First we make use of the implicit function theorem to derive
the partial derivatives of network consistent contract
values. Interestingly, to our knowledge this seems to be unknown to
much of the community -- a notable exception being the thesis of
\cite{Karl2015}. Second, we do not consider the ex-post values,
i.e. at the time of contract maturity/settlement, but instead consider
the ex-ante market values under the risk-neutral measure. In turn, we
compute several first-order Greeks and investigate their behavior in
several examples. In particular, we propose the network $\Delta$ as a
principled, quantitative measure of systemic risk.

Our paper is structured as follows. Section \ref{sec:Model} introduces
the mathematical notation and network valuation model. In section
\ref{sec:Greeks} we explain risk-neutral pricing and derive a formal
solution for the network Greeks. We illustrate our results on some
examples in section \ref{sec:illu}. Finally, we discuss implications
of our model and provide an outlook on future extensions in section
\ref{sec:discuss}. Major proofs and computations are collected in the
appendices. There we also relate our results to previous studies,
namely the {\em threat index} of \cite{Demange2018} and the local
valuation framework of \cite{Barucca2016}.

\section{Model}
\label{sec:Model}

\subsection{Notation and mathematical preliminaries}

Here we quickly summarize the mathematical notation employed in this
paper. We write vectors $\v{x}, \v{y} \in \mathbb{R}^n$ with bold
lower case and matrices $\v{A}, \v{B} \in \mathbb{R}^{m \times n}$
with bold upper case letters. Individual entries of vectors and
matrices are written as $x_i, A_{ij}$. $\diag(\v{x})$ denotes the
$n \times n$ diagonal matrix $D$ with entries $D_{ii} = x_i$ along its
diagonal. The transpose of a matrix is denoted as $\v{A}^T$. All
products containing vectors and matrices are understood as standard
matrix products, e.g. $\v{A} \v{B}$ denotes the matrix product of
$\v{A}$ and $\v{B}$, $\v{x} \v{x}$ is undefined whereas
$\v{x}^T \v{x}$ is the scalar product of $\v{x}$ with itself. Row- and
column-wise stacking of vectors or matrices is denoted by
$(\v{x}; \v{y})$ and $(\v{x}, \v{y})$ respectively, i.e.
$(\v{x}; \v{y})$ is a $2 n$-dimensional vectors whereas
$(\v{x}, \v{y})$ is a $n \times 2$ matrix.

Random variables $X, Y$ are written as upper case letters with
individual outcomes $x, y$ denoted in lower case. Expectations are
denoted as $\E[f(X)]$ and understood with respect to the (joint)
distribution of random variables within the brackets. Sometimes we
use $\E_t^Q$ to denote that the expectation is taken over the
risk-neutral measure $Q$, implicitly conditioned on the information
filtration upto time $t$.

\subsection{Network valuation}

\cite{Merton1974} has shown that equity and firm debt can be
considered as call and put options on the firm's value
respectively. In this model, a single firm is holding externally
priced assets $a$ and zero-coupon debt with nominal amount $d$ due at
a single, fixed maturity $T$. Then, at time $T$ the value of equity
$s$ and the recovery value of debt $r$ are given as
\begin{align}
  \label{eq:Merton}
  s &= \max\{0, a - d\} = (a - d)^{+}, \\
  r &= \min\{d, a\} = d - (d - a)^{+}
\end{align}
corresponding to an implicit call and put option respectively.

\cite{Suzuki2002} and others \citep{Elsinger2009,Fischer2014} have since
generalized this model to multiple firms with equity and debt
cross-holdings. In this paper we consider $n$ firms. Each firm
$i = 1, \ldots, n$ holds an external asset $a_i > 0$ as well as a fraction
$M_{ij}^s$ of firm $j$'s equity and debt $M_{ij}^d$. Here, the
investment fractions $M_{ij}^s$ and $M_{ij}^d$ are bounded between $0$
and $1$, i.e. $0 \leq M_{ij}^{s,d} \leq 1$, and the actual value
invested in the equity of counterparty $j$ is given as $M_{ij}^s
s_j$. In the following we require:
\begin{assumption}
  \label{assu_1}
  There are no self-holdings, i.e. $M^s_{ii} = M^d_{ii} = 0$ for all
  $i = 1, \ldots, n$, nor short positions, i.e.
  $M^s_{ij}, M^d_{ij} \geq 0$ for all $i, j = 1, \ldots, n$. Moreover,
  the fraction of equity and debt held by any counterparty cannot
  exceed unity, i.e. for all $j = 1, \ldots, n$ we require that
  \begin{align}
    \label{eq:assu_1a}
    \sum_i M_{ij}^s \leq 1 \quad \mbox{and} \quad \sum_i M_{ij}^d \leq 1 \, .
  \end{align}
  Furthermore, some equity and debt are held externally, i.e. there
  exist firms $j_s$ and $j_d$ such that
  \begin{align}
    \label{eq:assu_1b}
    \sum_i M_{ij_s}^s < 1 \quad \mbox{and} \quad \sum_i M_{ij_d}^d < 1 \, .
  \end{align}  
\end{assumption}
That is, $\v{M}^s$ and $\v{M}^d$ are strictly (left) sub-stochastic
matrices.

Now, the value of all assets $v_i$ held by firm $i$ is given by
\begin{align}
  \label{eq:XOS_val_i}
  v_i &= a_i + \sum_{j=1}^n M_{ij}^s s_j + \sum_{j=1}^n M_{ij}^d r_j \, .
\end{align}
Correspondingly, the firm's equity and recovery value of debt are given by
\begin{align}
  \label{eq:XOS_sr_i}
  s_i &= \max\left\{0, a_i + \sum_j M^s_{ij} s_j + \sum_j M^d_{ij} r_j - d_i \right\}, \\
  r_i &= \min\left\{d_i, a_i + \sum_j M^s_{ij} s_j + M^d_{ij} r_j \right\} \, .
\end{align}
In matrix notation, i.e. collecting equity and debt values into
vectors $\v{s} = (s_1, \ldots, s_n)^T$ and
$\v{r} = (r_1, \ldots, r_n)^T$ respectively, this can be rewritten as
\begin{align}
  \label{eq:XOS_vec}
  \v{s} &= \max\left\{\v{0}, \v{a} + \v{M}^s \v{s} + \v{M}^d \v{r} - \v{d} \right\}, \\
  \v{r} &= \min\left\{\v{d}, \v{a} + \v{M}^s \v{s} + \v{M}^d \v{r} \right\}
\end{align}
Thus, the firms' equity and debt values are endogenously defined as
the solution of a fixed point. This is readily seen when collecting
equity and debt row-wise into a single vector $\v{x} = (\v{s}; \v{r})$,
i.e. $\v{s} = \v{x}_{1:n}$ and $\v{r} = \v{x}_{(n+1):2n}$, and writing
\begin{align}
  \label{eq:XOS_fix}
  \v{x} = \v{g}(\v{a}, \v{x})
\end{align}
with the vector valued function
$\v{g} = (g^s_1, \ldots, g^s_n, g^r_1, \ldots, g_n^r)^T$ where for
$i = 1, \ldots, n$
\begin{align}
  \label{eq:XOS_g}
  g^s_i(\v{a}, \v{x}) &= \max\left\{0, a_i + \sum_j M^s_{ij} x_{j} + \sum_j M^d_{ij} x_{n + j} - d_i \right\}, \\
  g^r_i(\v{a}, \v{x}) &= \min\left\{d_i, a_i + \sum_j M^s_{ij} x_{j} + \sum_j M^d_{ij} x_{n + j} \right\} \, . \\
\end{align}
Each of the functions $g_i^s$ and $g_i^r$ is continuous and increasing
in $\v{a}$ and $\v{x}$. Together with assumption \ref{assu_1} it
follows that the fixed point of \eq{XOS_fix} is positive and unique.
\begin{theorem}
  Suppose that assumption \ref{assu_1} holds. Then, for each value of
  external assets $\v{a} > \v{0}$ there is a positive and unique
  $\v{x}$ solving \eq{XOS_fix}.
\end{theorem}
\begin{proof}
  Our model is a special case of the one considered by
  \cite{Fischer2014} with $k = 1$ and $\v{d}^1_{\v{r}^1,\v{r}^0}
  \equiv \v{d}$. Furthermore, Fischer's assumption 3.1 holds by
  assumption \ref{assu_1} and assumptions 3.6 and 3.7 are trivial as
  our nominal debt vector $\v{d}$ is constant. The result then follows
  by his theorem 3.8 (iv). 
\end{proof}

It is interesting to consider the model from the perspective of an
external investor. Whereas the internal value of all firms is given by
$\v{v} = \v{s} + \v{r}$, stock and debt holdings are diluted by
cross-holdings.
Considering firm $i$, a fraction $\sum_j M_{ji}^s$ of its stocks is
held by other firms $j$.  Thus only a fraction $1 - \sum_j M_{ji}^s$
is available to outside investors and similarly for debt.
\begin{definition}
  The value $v^{\mathrm{out}}_i$ of firm $i$ available to outside
  investors is defined as
  \begin{align}
    \label{eq:v_ext}
    v_i^{\mathrm{out}} &= (1 - \sum_j M_{ji}^s) s_i + (1 - \sum_j M_{ji}^d) r_i \, .
  \end{align}  
\end{definition}
While the internal value exceeds the value of external assets, i.e.
\begin{align}
  \label{eq:val}
    \v{v} &= \v{s} + \v{r} \\
    &= \max\left\{\v{0}, \v{a} + \v{M}^s \v{s} + \v{M}^d \v{r} - \v{d} \right\} + \min\left\{\v{d}, \v{a} + \v{M}^s \v{s} + \v{M}^d \v{r} \right\} \\
    &= \v{a} + \v{M}^s \v{s} + \v{M}^d \v{r} > \v{a} ,
\end{align}
this is not the case for external investors.

\begin{proposition}
  \label{prop:out}
  The total value available to outside investors equals the total
  value of external assets, i.e.
  $\sum_{i=1}^n v_{i}^{\mathrm{out}} = \sum_{i=1}^n a_i$.
\end{proposition}
\begin{proof}
  From \eq{val} the total value of all firms is given as
\begin{align}
  \sum_i v_i &= \sum_i \left( s_i + r_i \right) = \sum_i \left( a_i + \sum_j M_{ij}^s s_j + \sum_j M_{ij}^d r_j \right)
\end{align}
Outside investors hold a total value of
\begin{align}
  \sum_i v_i^{\mathrm{out}}
  &= \sum_i (1 - \sum_j M_{ji}^s) s_i + \sum_i (1 - \sum_j M_{ji}^d) r_i \\
  &= \sum_i \left( s_i + r_i \right) - \sum_i \sum_j M_{ji}^s s_i - \sum_i \sum_j M_{ji}^d r_i \\
  &= \sum_i \left( a_i + \sum_j M_{ij}^s s_j + \sum_j M_{ij}^d r_j \right) - \sum_j \sum_i M_{ij}^s s_j - \sum_j \sum_i M_{ij}^d r_j \\
  &= \sum_i a_i
\end{align}
\end{proof}
This point, that asset values are inflated by cross-holdings compared
to the actual underlying external values accessible by investors, has
also been noted by \cite{Fischer2014}. In general,
$v_i \neq v_i^{\mathrm{out}}$ and thus, it can be beneficial or
detremential for an outside investor, if firm $i$ enters into
cross-holding agreements with counterparties. Indeed, as external
asset value is conserved in the model, any change in the structure of
cross-holdings gives rise to a zero-sum redistribution between
external investors.

\section{Network Greeks}
\label{sec:Greeks}

How is risk distributed in and by the network of cross-holdings? Risk
in financial derivatives is routinely measured and management using
the Greeks or risk sensitivities \cite{Haug2003a,Haug2003b}. In this section, we
build on the above model which considers the value of network claims
as an extension of Merton's model. Accordingly, firm equity and debt
values under cross-holdings are derivative contracts and can be priced
and managed as such. In particular, we compute ``network Greeks''
to investigate systemic risk arising from interconnected financial
firms. To our knowledge this idea has not been investigated before.

\subsection{Risk-neutral valuation}

Denoting the unique solution of \eq{XOS_fix} by $\v{x}^*(\v{a})$, we
can consider the corresponding value of equity and debt claims as
derivative contracts on the underlying $\v{a}$. Accordingly, the
ex-ante market price at time $t < T$ is given as
\begin{align}
  \label{eq:val_Q}
  \v{x}_t = \E_t^Q[e^{-r \tau} \v{x}^*(\v{A}_T)]
\end{align}
with the riskless interest rate $r$ and time to maturity
$\tau = T - t$. The expectation is taken with respect to the
risk-neutral measure $Q$ of external asset values $\v{a}$ at maturity
$T$. In the following, we assume that the risk-neutral asset values
follow a multi-variate geometric Brownian motion, i.e.
\begin{align}
  \label{eq:SDE}
  dA^i_t &= r A^i_t \d{t} + \sigma_i A^i_t \d{W^i_t}
\end{align}
with possibly correlated Wiener processes $W^i_t$, i.e.
$\E[\d{W^i_t} \d{W^j_t}] = \rho_{ij} \d{t}$ with $\rho_{i,i} = 1$. It is
well-known that the solution $\v{A}_t$ is given by
\begin{align}
  \label{eq:p_A_t}
  A^i_t &= a^i_0 e^{\left( r - \frac{1}{2} \sigma_i^2 \right) t + \sigma_i W^i_t}
\end{align}
where $a^i_0 > 0$ denotes the initial value and $\v{W}_t$ is
multivariate normal distributed with mean $\v{0}$ and covariance
matrix $t \v{C}$ with entries $C_{ij} = \rho_{ij}$. Note that $\v{W}_t$ can
be obtained from independent standard normal variates
$\v{Z} \sim \mathcal{N}(\v{0}, \v{I}_{n \times n})$ as $\v{W}_t = \sqrt{t} \v{L} \v{Z}$ with
$\v{L}^T \v{L} = \v{C}$. We will use this
representation in the next section to express the risk-neutral market value of equity and
debt contracts as
\begin{align}
  \label{eq:val_Q}
  \v{x}_t = \E_t^Q[e^{-r \tau} \v{x}^*\left(\v{a}_{T}(\v{Z})\right)] 
  = \E_t^Q\left[e^{-r \tau} \v{x}^*\left(\v{a}_t e^{\left( r - \frac{1}{2} \v{\sigma}^2 \right) \tau + \sqrt{\tau} \diag(\v{\sigma}) \v{L} \v{Z}}\right)\right] \, .
\end{align}

\subsection{Formal solution}

The Greeks quantify the sensitivities of derivative prices to changes
in underlying parameters. As a starting point, in this paper, we
consider first-order Greeks only. In particular, we investigate the
sensitivity of equity and debt prices accounting for cross-holdings
with respect to current asset values
$\Delta = \frac{\partial \v{x}_t}{\partial \v{a}_t}$, volatilities
$\mathcal{V} = \frac{\partial \v{x}_t}{\partial \v{\sigma}}$,
risk-free interest rate $\rho = \frac{\partial \v{x}_t}{\partial r}$
and time to maturity $\Theta = - \frac{\partial \v{x}_t}{\partial \tau}$.

Denoting all parameters of interest by
$\v{\theta} = (\v{a}_t, \v{\sigma}, r, \tau)^T$ and considering that the
asset value $\v{a}_{\tau}(Z; \theta)$ depends on the random variate
$Z$ and these parameters, we need to compute the following derivatives
\begin{align}
  \label{eq:greek_Q}
  \frac{\partial}{\partial \v{\theta}} \v{x}_t &= \frac{\partial}{\partial \v{\theta}} \E_t^Q[ e^{- r \tau} \v{x}^*(\v{a}_{T}(Z; \theta)) ] \\
  &= \E_t^Q\left[ \left( \frac{\partial}{\partial \v{\theta}} e^{- r \tau} \right) \v{x}^*(\v{a}_{T}(Z; \theta))
    + e^{- r \tau} \left(\frac{\partial}{\partial \v{\theta}} \v{x}^*(\v{a}_{T}(Z; \theta)) \right) \right]
\end{align}
where we have used Leibniz's rule to exchange the order of expectation and
differentiation.

By the chain rule of differentiation we obtain
\begin{align}
  \label{eq:greek_chain}
  \frac{\partial}{\partial \v{\theta}} \v{x}^*(\v{a}_{\tau}(Z; \theta)) &= \frac{\partial}{\partial \v{a}} \v{x}^*(\v{a}) \mid_{\v{a} = \v{a}_{\tau}(Z; \theta)} \frac{\partial}{\partial \v{\theta}} \v{a}_{\tau}(Z; \theta) \, .
\end{align}
Note that $\frac{\partial}{\partial \v{a}} \v{x}^*(\v{a})$ is the
derivative of the fixed point solving \eq{XOS_fix}. In order to
compute it, we make use of the {\em implicit function theorem}.
A version of the theorem by \cite{Halkin1974} is adopted to our notation:
\begin{theorem}
  \label{thm:impl_fun}
  Let $U \subset \mathbb{R}^m, V \subset \mathbb{R}^n$
  and $\v{f}: U \times V \to \mathbb{R}^n$ a continuously
  differentiable function. Suppose that
  \begin{align}
    \label{eq:impl_fix}
    \v{f}(\v{x}^*, \v{y}^*) = \v{0}
  \end{align}
  at a point $(\v{x}^*, \v{y}^*) \in U \times V$ and that the Jacobian
  matrices
  $\v{J}_{\v{f}, \v{x}} \v{f}(\v{x}, \v{y}), \v{J}_{\v{f}, \v{y}}
  \v{f}(\v{x}, \v{y})$
  of partial derivatives exist at $(\v{x}^*, \v{y}^*)$. Further, $\v{J}_{\v{f}, \v{y}}$ is invertible at
  this point. Then, there exists a neighbourhood
  $U^* \subset U$ and a 
  continuously differentiable function $\v{h}: U^* \to \mathbb{R}^n$
  with
  \begin{align}
    \label{eq:impl_sol}
    \v{h}(\v{x}^*) &= \v{y}^*
  \end{align}
  and
  \begin{align}
    \label{eq:impl_sol_fix}
    \v{f}(\v{x}, \v{h}(\v{x})) = \v{0} \quad \forall \v{x} \in U^* \, .
  \end{align}
  Moreover, the partial derivatives of $\v{h}$ with respect to
  $\v{x} \in U^*$ are given as
  \begin{align}
    \label{eq:impl_sol_deriv}
    \frac{\partial}{\partial \v{x}} \v{h}(\v{x}) &= - \left[ \v{J}_{\v{f}, \v{y}} \v{f}(\v{x}, \v{h}(\v{x})) \right]_{n \times n}^{-1}
                                                   \left[ \frac{\partial}{\partial \v{x}} \v{f}(\v{x}, \v{h}(\v{x})) \right]_{n \times m}
  \end{align}
\end{theorem}
Note that the function $\v{g}(\v{a}, \v{x})$ defined in \eq{XOS_g} is
continuous and almost everywhere differentiable. The partial
derivatives are given by
\begin{align}
  \label{eq:g_diff}
  \frac{\partial}{\partial s_j} g_i^s(\v{a}, \v{x}) &= 
                                                        \left\{ \begin{array}{rp{1cm}l} M^s_{ij} & & \mbox{if firm $i$ is solvent} \\
                                                                  0 & & \mbox{otherwise} \end{array} \right. \\
  \frac{\partial}{\partial s_j} g_i^r(\v{a}, \v{x}) &= 
                                                        \left\{ \begin{array}{rp{1cm}l} 0 & & \mbox{if firm $i$ is solvent} \\
                                                                  M^s_{ij} & & \mbox{otherwise} \end{array} \right. \\
  \frac{\partial}{\partial r_j} g_i^s(\v{a}, \v{x}) &= 
                                                        \left\{ \begin{array}{rp{1cm}l} M^d_{ij} & & \mbox{if firm $i$ is solvent} \\
                                                                  0 & & \mbox{otherwise} \end{array} \right. \\
  \frac{\partial}{\partial r_j} g_i^r(\v{a}, \v{x}) &= 
                                                        \left\{ \begin{array}{rp{1cm}l} 0 & & \mbox{if firm $i$ is solvent} \\
                                                                  M^d_{ij} & & \mbox{otherwise} \end{array} \right. \\
  \frac{\partial}{\partial a_j} g_i^s(\v{a}, \v{x}) &= 
                                                        \left\{ \begin{array}{rp{1cm}l} \phantom{M} 1 & & \mbox{if $i = j$ and firm $i$ is solvent} \\
                                                                  0 & & \mbox{otherwise} \end{array} \right. \\
  \frac{\partial}{\partial a_j} g_i^r(\v{a}, \v{x}) &= 
                                                        \left\{ \begin{array}{rp{1cm}l} \phantom{M} 0 & & \mbox{if $i = j$ and firm $i$ is solvent} \\
                                                                  1 & & \mbox{otherwise} \end{array} \right. \, .
\end{align}
Here, a firm $i$ is solvent if its asset value $v_i$ is sufficient to
repay its nominal debt $d_i$, i.e.
$v_i = a_i + \sum_{j=1}^n M_{ij}^s s_j + \sum_{j=1}^n M_{ij}^d r_j >
d_i$.
The derivatives of $\v{g}$ exist everywhere except for the boundary
case $v_i = d_i$. Defining the solvency vector
$\v{\xi} = (\mathbbm{1}_{v_i > d_1}(v_1), \ldots, \mathbbm{1}_{v_n >
  d_n}(v_n))$,
the partial derivatives of $\v{g}$ with respect to $\v{x}$ can be
collected in a matrix as follows
\begin{align}
  \label{eq:g_diff_mat}
  \frac{\partial}{\partial \v{x}} \v{g}(\v{a}, \v{x}) 
  &= \left[
    \begin{array}{ccc}
      \diag(\v{\xi}) \v{M}^s & & \diag(\v{\xi}) \v{M}^d \\
      & & \\
      \diag(\v{1}_n - \v{\xi}) \v{M}^s &  & \diag(\v{1}_n - \v{\xi}) \v{M}^d
    \end{array}
    \right] \\
  &= \diag\left((\v{\xi}; \v{1}_n - \v{\xi})\right) \left[
    \begin{array}{ccc}
      \v{M}^s & & \v{M}^d \\
      & & \\
      \v{M}^s &  & \v{M}^d
    \end{array}
    \right]                                 
\end{align}

Thus, defining $\v{f}(\v{a}, \v{x}) = \v{x} - \v{g}(\v{a}, \v{x})$ we
obtain by the implicit function theorem \ref{thm:impl_fun}
\begin{corollary}
  \label{coro:fix_diff}
  The partial derivatives of $\v{x}^*(\v{a})$ are given by
  \begin{align}
    \label{eq:fix_diff}
    \frac{\partial}{\partial \v{a}} \v{x}^*(\v{a}) 
    &= \left[ \v{I}_{2 n \times 2 n} - \frac{\partial}{\partial \v{x}} \v{g}(\v{a}, \v{x}) \right]^{-1}
      \left[ \begin{array}{c} \diag(\v{\xi}) \\ \\ \diag(\v{1}_n - \v{\xi}) \end{array} \right]
  \end{align}
\end{corollary}
\begin{proof}
  Use that
$\frac{\partial}{\partial \v{x}} \v{f}(\v{a}, \v{x}) = \v{I}_{2 n
  \times 2 n} - \frac{\partial}{\partial \v{x}} \v{g}(\v{a}, \v{x})$
and
$\frac{\partial}{\partial \v{a}} \v{f}(\v{a}, \v{x}) = -
\frac{\partial}{\partial \v{a}} \v{g}(\v{a}, \v{x})$.
Then, the result follows from theorem \ref{thm:impl_fun} and
$\frac{\partial}{\partial \v{a}} \v{g}(\v{a}, \v{x}) =
\left[ \begin{array}{c} \diag(\v{\xi}) \\ \diag(\v{1}_n -
    \v{\xi}) \end{array} \right]$.
As explained below, assumption \ref{assu_1} ensures that
$\frac{\partial}{\partial \v{x}} \v{f}(\v{a}, \v{x})$ is invertible as
required. 
\end{proof}

Finally, combining \eq{greek_Q} and (\ref{eq:greek_chain}) with corollary
\ref{coro:fix_diff} we formally compute the network Greeks as
\begin{align}
  \label{eq:greek_net}
  \frac{\partial}{\partial \v{\theta}} \v{x}_t
  &= \E_t^Q\left[ \left( \frac{\partial}{\partial \v{\theta}} e^{- r \tau} \right) \v{x}^*(\v{a}_{T}(Z; \theta)) \right. \nonumber \\
  &\phantom{=} \left. + e^{- r \tau} \left[ \v{I}_{2 n \times 2 n} - \frac{\partial}{\partial \v{x}} \v{g}(\v{a}, \v{x}) \right]^{-1}
      \left[ \begin{array}{c} \diag(\v{\xi}) \\ \\ \diag(\v{1}_n - \v{\xi}) \end{array} \right]
  \frac{\partial}{\partial \v{\theta}} \v{a}_{T}(Z; \theta) \right]
\end{align}
where the expectation is well-defined as the derivatives exist almost
everywhere, i.e. except for a set of measure zero.

Note that the effect of the cross-holding network is fully captured by
the matrix
\begin{align}
  \label{eq:net_inv}
  \v{W}
  &= \left[ \v{I}_{2 n \times 2 n} - \frac{\partial}{\partial \v{x}}
    \v{g}(\v{a}, \v{x}) \right]^{-1}
  = \left[ \v{I}_{2 n \times 2 n} - \diag\left( (\v{\xi}; \v{1}_n - \v{\xi}) \right)
    \left[
    \begin{array}{ccc}
      \v{M}^s & & \v{M}^d \\
      & & \\
      \v{M}^s &  & \v{M}^d
    \end{array}
                   \right] \right]^{-1}
\end{align}
weighting the contributions from the direct sensitivities
$\frac{\partial}{\partial \v{\theta}} \v{a}_{\tau}(Z; \theta)$. In
case of no cross-holdings, $\v{M}^s, \v{M}^d \equiv \v{0}$, this
reduces to the identity matrix by \eq{g_diff_mat}. 

The network effect captured by the weighting matrix $\v{W}$ has some interesting
implications:
\begin{enumerate}
\item It can be considered as a Katz-Bonacich type centrality, i.e. of
  the form $(\v{I} - \alpha \v{M})^{-1}$. In our case, the weighting
  factor $\alpha$ is different for different firms and endogenously
  derived from their solvency status $\v{\xi} \in \{0, 1\}$. Here, we note that
  Katz-Bonacich centrality exists as long as the largest eigenvalue of
  $\alpha \v{M}$ is below one. Thus, from this connection we obtain
  that the inverse in \eq{net_inv} exists, as
  $\v{\xi} \in \{0, 1\}^n$ and we assumed sub-stochastic cross-holding
  matrices $\v{M}^s, \v{M}^d$ by assumption \ref{assu_1}. Furthermore, it opens up the
  possibility for targeted interventions striving to optimally
  control node centralities \citep{Reiffersmasson2015}.

  Moreover, using the expansion
  \begin{align}
    (\v{I} - \alpha \v{M})^{-1} &= \sum_{k=0}^{\infty} \left(\alpha
    \v{M}\right)^k = \v{I} + \alpha \v{M} + \left( \alpha \v{M}
    \right)^2 + \ldots    
  \end{align}
  it is clear that the network always amplifies initial shocks.
\item The effective network changes based on the distance to default of
  firms. Consider the case when all firms are solvent, i.e.
  $\v{\xi} = \v{1}$. Then, from corollary \ref{coro:fix_diff} the
  sensitivity of equity and debt values to changing asset prices are
  given as
  \begin{align}
    \label{eq:diff_solvent}
    \frac{\partial}{\partial \v{a}} \v{x}^*(\v{a}) 
    &= \left[ \begin{array}{c} (\v{I}_{n \times n} - \v{M}^s)^{-1} \\ \v{0} \end{array} \right] \, .
  \end{align}
  Thus, only the cross-holdings of equity are visible and debt values
  are uneffected as all contracts are honored at their nominal
  values. In contrast, when all firms are insolvent, i.e.
  $\v{\xi} = \v{0}$, the sensitivities are given by
  \begin{align}
    \label{eq:diff_insolvent}
    \frac{\partial}{\partial \v{a}} \v{x}^*(\v{a}) 
    &= \left[ \begin{array}{c} \v{0} \\ (\v{I}_{n \times n} - \v{M}^d)^{-1} \end{array} \right] \, .
  \end{align}
  with only debt cross-holdings visible. This implies that the Greeks
  of market values, being averages as of \eq{greek_net}, can change
  substantially if the default probability of companies
  changes. E.g. debt cross-holdings might be almost invisible under
  normal circumstances, yet drive the risk sensitivities in a crisis
  when default probabilities rise. Suggesting that systemic risk
  management must take into account stressed scenarios, i.e. by
  elevated default probabilities, to be meaningful in a crisis.

  Based on this observation, we hypothesize that these extreme cases
  correspond to the largest sensitivities of equity and debt
  respectively.
  \begin{hypothesis}
    \label{hyp:mono}
    Let assumption \ref{assu_1} hold and split the partial derivatives
    of $\v{x}^*(\v{a})$ as
    \begin{align}
      \left[ \begin{array}{c} \v{u}^s \\ \v{u}^d \end{array} \right] 
   &= \frac{\partial}{\partial \v{a}} \v{x}^*(\v{a}) \, .
    \end{align}
    Then, $\v{u}^s$ is monotonically increasing and $\v{u}^d$ is
    monotonically decreasing when considered as a function of $\v{\xi}$.
  \end{hypothesis}
  A proof under slightly stronger assumptions can be found in appendix
  \ref{app:hyp}.
\item In case of cross-holdings of debt only, \cite{Demange2018} has
  proposed a {\em threat index} $\v{\mu}$ measuring the spill-over
  potential of each firm. In appendix \ref{app:Demange} we derive the
  exact relation with our model and show how the index $\v{\mu}$ can be
  computed from the partial derivatives derived in corollary
  \ref{coro:fix_diff}.

  Based on the idea of a threat index, one might be tempted to
  consider the expectation
  $\E_t^Q[\frac{\partial}{\partial A_T} \v{x}^*(\v{A_T})]$ as a
  measure of firms' spill-overs towards each other. In particular, the
  total impact $\pi_i$ of a change to the external asset of firm
  $i$ on the value of all other firms might be measured as
  \begin{align}
    \label{eq:impact}
    \v{\pi} &= \E_t^Q[ \sum_i \frac{\partial v_i}{\partial \v{A}_T} ] \\
    &= \E_t^Q[\v{1}^T \frac{\partial}{\partial \v{A}A_T} \v{x}^*(\v{A}_T)]^T  \\
    &= \E_t^Q\left[\v{1}^T \v{W} (\diag(\v{\xi}); \diag(\v{1}_n - \v{\xi}))\right]^T \, .
  \end{align}
  While $\pi_i$ might provide a useful proxy for a firm's systemic
  risk, the Greeks in general consider a firm's risk differently. In
  particular, the total risk sensitivity of all banks with respect to
  the parameters $\v{\theta}$ cannot be computed as
  \begin{align}
    \sum_i \frac{\partial}{\partial \v{\theta}} (\v{x}_t)_i
    &\neq
    \v{1}^T \E_t^Q\left[ \left( \frac{\partial}{\partial \v{\theta}} e^{- r \tau} \right) \v{x}^*(\v{a}_{T}(Z; \theta)) \right]
    + e^{- r \tau} \v{\pi}^T
    \E_t^Q\left[ \frac{\partial}{\partial \v{\theta}} \v{a}_{T}(Z; \theta) \right]
  \end{align}
  as $\v{\xi}$ and
  $\frac{\partial}{\partial \v{\theta}} \v{a}_{\tau}(Z; \theta)$ both
  depend on $Z$ and are generally not independent.

  Note that $\v{\pi}$ can be interpreted as the aggregate impact, i.e. on
  all firms, of an asset price shock at maturity. Similarly, 
  $\v{1}^T \Delta = \v{1}^T \E_t^Q[e^{- r \tau} \frac{\partial}{\partial \v{a}_t} \v{x}^*(\v{A_T})]$
  quantifies the aggregate impact of asset price shocks at current market
  values. As this is arguably more relevant for risk management purposes, we
  propose $\Delta^{\text{Total}} = \v{1}^T \Delta$ as a better suited
  measure of systemic risk. In the examples in section \ref{sec:illu}
  we illustrate and comment in more detail on their differences.
\item Risk management is also interesting from the perspective of
  outside investors. In this case, from proposition \ref{prop:out} we
  obtain that any risk is redistributed between them without
  amplification as compared to the risk on external assets directly:
  \begin{align}
    \frac{\partial}{\partial \v{\theta}} \sum_{i=1}^n v_{i}^{\mathrm{out}}
    &= \sum_{i=1}^n \frac{\partial}{\partial \v{\theta}} a_i
  \end{align}
  In particular, for the partial derivatives with respect to the
  external assets $\v{a}$, we find that
  $\sum_i \frac{\partial}{\partial \v{a}} v_i^{\mathrm{out}} = n$,
  i.e. on average each investor bears a risk of $1$ as he would when
  holding a single external asset directly. Yet, in general the risk
  is redistributed between outside investors, suggesting to view the
  individual $\frac{\partial v_{i}^{\mathrm{out}}}{\partial \v{a}}$
  vectors as the weights of the implicit portfolio that the investor
  holds in external assets.
\end{enumerate}

\subsection{Local approximation}

\cite{Barucca2016} have compared several network contagion models and
unified them in terms of network valuations solving for a
self-consistent fixed point solution. In particular, they showed that
this includes many existing models, including the seminal model by
\cite{Eisenberg2001} and the well-known DebtRank contagion process
\citep{Battiston2012}, valuing assets ex-post at maturity. In the case
of ex-ante valuations, i.e. at current market prices, they proposed a
{\em local approximation} The idea being that each firm evaluates its
portfolio locally, i.e. by pricing external assets and holdings of
other companies at market prices. 

Here, we discuss how this approximation relates to our setup. For
simplicity we consider the case of debt cross-holdings only. Then,
the market value of debt of firm $i$ solves
\begin{align}
  r_i &= \min\{ d_i, a_i + \sum_j M_{ij}^d r_j \} \, .
\end{align}
This can be written in the form of valuation defined by
\cite{Barucca2016} as 
\begin{align}
  r_i &= d_i \mathbbm{1}_{E_i > 0} + (E_i + d_i)^{+} \mathbbm{1}_{E_i \leq 0}
\end{align}
with $E_i = a_i + \sum_{j} M_{ij}^d r_j - d_i = v_i - d_i$.
Using the valuation functions of \cite{Barucca2016}
\begin{align}
  \mathbb{V}^e(E_i) &= 1 \\
  \mathbb{V}^e_{ij}(E_j) &= \mathbbm{1}_{E_j > 0} + \left( \frac{E_j + d_j}{d_j} \right)^+ \mathbbm{1}_{E_j \leq 0}
\end{align}
the market value of equity reads
\begin{align}
  \label{eq:local_fix}
  E_i &= a_i \mathbb{V}^e(\v{E}) + \sum_j M_{ij}^d d_j \mathbb{V}_{ij}(\v{E}) - d_i \, .
\end{align}
Notice that in our setup the total liability of firm $i$ is given by
$d_i$ and a fraction $\sum_{j} M_{ji}^d$ of it is payed to other
firms. Further, the nominal amount lend by $i$ to $j$ is given as
$M_{ij}^d d_j$. Thus \eq{local_fix} corresponds to equation (7) of
\cite{Barucca2016}.

Finally, using an ex-ante valuation, e.g. based on the Black-Scholes
formula, to compute market values at time $t$ when nominal debt
payments are due at maturity $T$, we obtain
\begin{align}
  \mathbb{V}^e(E_i(t)) &= 1 \\
  \mathbb{V}^e_{ij}(E_j(t)) &= 
    \E^Q_t\left[ \mathbbm{1}_{E_j > 0} + \left( \frac{E_j + d_j}{d_j} \right)^+ \mathbbm{1}_{E_j \leq 0} \right] \\
                       &= \E^Q_t\left[ \mathbbm{1}_{E_j(T) > 0} \right] + \E^Q_t\left[ \frac{v_i}{d_i} \mathbbm{1}_{E_j(T) \leq 0} \right] \, .
\end{align}
Thus, the nominal holding $M_{ij}^d d_j$ is adjusted by the
risk-neutral probability of non-default
$\E^Q_t[ \mathbbm{1}_{E_j(T) > 0} ]$ and the expected shortfall
$\E^Q_t[ \frac{v_i}{d_i} \mathbbm{1}_{E_j(T) \leq 0} ]$. Another
interpretation sees the resulting market price of debt as its nominal
value insured by a short Black-Scholes put option, i.e.
\begin{align}
  d_j \left( \mathbbm{1}_{E_j > 0} + \frac{v_j}{d_j} \mathbbm{1}_{E_j \leq 0} \right)
  &= d_j - \left( d_j - v_j \right)^+
\end{align}
as $E_j \leq 0 \Leftrightarrow v_j \leq d_j$. The ex-ante valuation of
firm $i$th equity can thus be written as
\begin{align}
  \label{eq:local_put}
    E_i(t) &= a_i(t) + \sum_j M_{ij}^d \left( d_j - c_{\mathrm{put}}(E_j(t) + d_j, d_j) \right) - d_i
\end{align}
where $c_{\mathrm{put}}(v_j, d_j)$ denotes the Black-Scholes price of
a put option on the value of firm $j$ with ex-ante value $v_j$ and
strike $d_j$.  Notice that this includes additional approximations
besides the standard Black-Scholes assumptions. First, since $v_j$
denotes the value of a firms asset portfolio, equity and debt are
actually basket options and would need to be priced
accordingly. Further, the volatility of the portfolio value
$\sigma_{V_j}$ is assumed known and fixed instead of being an implicit
function of the volatility of the external asset $\sigma$ as in the
exact model. Secondly, further extensions such as strategic
default \citep{Leland1994} or roll-over risk \citep{He2012} are easily
included in this setup by plugging in the according pricing functions.

Denoting the fixed point of self-consistent market values by
$\v{E}^*(\v{a})$ and employing the inverse function theorem again it
is easy to show that the derivatives of the local approximation are
given by
\begin{align}
  \label{eq:local_Greeks}
  \frac{\partial}{\partial \v{a}} \v{E}^* &=
    \left( \v{I} - \v{M}^d \diag(\v{d}) \frac{\partial}{\partial \v{E}} \mathbb{V}(\v{E}^*) \right)^{-1} \v{I} \\
  &= \left( \v{I} + \v{M}^d \diag(\Delta_{\mathrm{put}}(\v{E}^* + \v{d}, \v{d})) \right)^{-1} \, ,
\end{align}
i.e. the cross-holding matrix is weighted by the $\Delta$s of the
local valuation functions. Interestingly, a very similar result has
been independently derived by \cite{Ota2014} who has considered the
propagation of asset price shocks when firms adjust the market values
of their portfolio holdings. \cite{Ota2014} introduced the notion of
{\em marginal contagion} and showed that an initial asset price shock
$\Delta\v{IA}^0$ is amplified by a network of debt cross-holdings as
\begin{align}
  \left( \v{I} - \v{M}^d \v{\Phi} \right)^{-1} \Delta\v{IA}^0 
\end{align}
where $\v{\Phi}$ is a diagonal matrix containing the risk-neutral
probabilities of default for each firm, i.e.
$\v{\Phi} = \diag(\v{1} - \E^Q_t[\v{\xi}])$ in our notation. Thus, the
cross-holdings $\v{M}^d$ are risk adjusted by the probabilities of
default instead of the valuation $\Delta$s as derived above for the
local approximation proposed by \cite{Barucca2016}.

When ignoring default correlations, we can provide yet another
connection with ex-ante value of the partial derivatives computed in
corollary \ref{coro:fix_diff}. Considering the case of debt
cross-holdings only, it suffices to consider the debt values
$\v{r}^*(\v{a})$. The partial derivatives are then given by
\begin{align}
  \frac{\partial}{\partial \v{a}} \v{r}^*(\v{a}) 
  &= \left( \v{I} - \diag(\v{1} - \v{\xi}) \v{M}^d \right)^{-1} \diag(\v{1} - \v{\xi})
\end{align}
Series expanding the matrix inverse, the risk-neutral ex-ante
value can be computed as
\begin{align}
  \E^Q_t[ \frac{\partial}{\partial \v{a}} \v{r}^*(\v{a}) ]
  &= \E^Q_t\left[ \sum_{k=0}^{\infty} \left( \diag(\v{1} - \v{\xi}) \v{M}^d \right)^k \diag(\v{1} - \v{\xi}) \right] \\
  &= \sum_{k=0}^{\infty} \E^Q_t\left[ \left( \diag(\v{1} - \v{\xi}) \v{M}^d \right)^k \diag(\v{1} - \v{\xi}) \right] \, .
\end{align}
In general, the expectation of matrix powers is difficult to
evaluate. But, assuming that defaults occur independently,
i.e. $\E^Q_t[\xi_i \xi_j] = \E^Q_t[\xi_i] \E^Q_t[\xi_j]$, the above
expression can be approximated as
\begin{align}
  \label{eq:local_indep}
  \E^Q_t[ \frac{\partial}{\partial \v{a}} \v{r}^*(\v{a}) ]
  &\approx \sum_{k=0}^{\infty} \E^Q_t[\diag(\v{1} - \v{\xi})] \left( \v{M}^d \right)^k \E^Q_t[\diag(\v{1} - \v{\xi})] \\
  &= \left( \v{I} - \diag(\v{1} - \E^Q_t[\v{\xi}]) \v{M}^d \right)^{-1} \diag(\v{1} - \E^Q_t[\v{\xi}]) \, .
\end{align}
This is most easily seen when noting that
\begin{align}
  \left( \diag(\v{1} - \v{\xi}) \v{M}^d \right)^2_{ij}
  &= \sum_k (1 -  \xi_i) M_{ik}^d (1 - \xi_k) M_{kj}^d \\
  \left( \diag(\v{1} - \v{\xi}) \v{M}^d \right)^3_{ij}
  &= \sum_{kl} (1 - \xi_i) M_{ik}^d (1 - \xi_k) M_{kl}^d (1 - \xi_l) M_{lj}^d, \ldots 
\end{align}
making it clear that the expectations over $\v{\xi}$ can be carried out
independently. 

This expression is very similar to the marginal contagion proposed by
\cite{Ota2014}. Yet, a major difference is that the risk adjustment by
the risk neutral default probabilities $\v{1} - \E^Q_t[\v{\xi}]$ is
done along the incoming instead of the outgoing connections. This
leads to different risk amplification as the matrices $\diag(\v{1} -
\E^Q_t[\v{\xi}]) \v{M}^d$ and $\v{M}^d \diag(\v{1} - \E^Q_t[\v{\xi}])$
are different in general. A possible interpretation could be, that each
firm adjusts the value of counterparties for their default probability
when managing its risk individually. Viewed from a network perspective,
it should instead adjust for its own default probability to account
for its contagion effect on other firms. Nevertheless,
\eq{local_indep} is also an approximation as \cite{Demange2018} has
shown that defaults are not independent, but positively correlated in
case of debt cross-holdings. Thus, by hypothesis \ref{hyp:mono} the
approximation underestimates the amplification of risk in financial
networks. With these remarks, we leave it to future work to explore
and understand the implications and differences of different local
approximations and turn to numerical illustraions of the network
Greeks.

\section{Numerical illustrations}
\label{sec:illu}

Here, we consider several examples illustrating our formal
solution. Starting with a symmetric, fully connected firm network
which can be solved analytically by being reduced to a single
representative firm, we then simulate the model with two firms and on
large random networks.

\subsection{Symmetric example}

As a first example, we consider $n$ identical firms. Each firm holds a
fractions $\frac{w^s}{n - 1}$ and $\frac{w^d}{n - 1}$ with parameters
$0 \leq w^s, w^d < 1$ of each counter parties equity and debt
respectively. Under these assumptions \eq{XOS_sr_i} simplifies to
\begin{align}
  s_i &= \max\left\{0, a_i + \sum_j \frac{w^s}{n - 1} s_j + \sum_j \frac{w^d}{n - 1} r_j - d_i \right\}, \\
  r_i &= \min\left\{d_i, a_i + \sum_j \frac{w^s}{n - 1} s_j + \sum_j \frac{w^d}{n - 1} r_j \right\} \, .
\end{align}
Furthermore, we consider a symmetric situation of identical firms, all
having the same nominal debt $d_i = d, \forall i$ and external asset
$a_i = a, \forall i$. Then, by symmetry $s_i = s$ and
$r_i = r, \forall i$, reducing the problem to a one-dimensional fixed
point
\begin{align}
  s &= \max\{0, a + w^s s + w^d r - d\} \\
  r &= \min\{d, a + w^s s + w^d r\} \, .
\end{align}
Now, consider two cases corresponding to zero and positive equity:
\begin{description}
\item[$s = 0$:] Assuming $r < d$ we obtain the recursion
  $r = a + w^s s + w^d r = a + w^d r$ with solution
  $r^* = \frac{a}{1 - w^d}$. This solution is consistent with our
  assumption as long as $r^* < d \Leftrightarrow a < (1 - w^d) d$.
\item[$s > 0$:] In this case, we can consistently assume
  $r = d$. Then, $s = a + w^s s + w^d d - d$ with
  solution $s^* = \frac{a - (1 - w^d) d}{1 - w^s}$. Indeed,
  $s^* > 0 \Leftrightarrow a > (1 - w^d) d$.
\end{description}
Thus, the two solution branches are mutually exclusive and connect at
the default boundary given by $a = (1 - w^d) d$. It is easily checked
that this is indeed the condition where the firms value equals its
nominal debt, i.e. $r^* + s^* = d$. Combining the two case, we obtain the solution
\begin{align}
  \label{eq:sol_one}
  s^* &= \left\{ \begin{array}{cl} \frac{a - (1 - w^d) d}{1 - w^s} & \mbox{ if } \xi = 1 \\
                 0 & \mbox{ otherwise} \end{array} \right. \\
  r^* &= \left\{ \begin{array}{cl} \phantom{xxx} d \phantom{xxx} & \mbox{ if } \xi = 1 \\
                 \frac{a}{1 - w^d} & \mbox{ otherwise} \end{array} \right. ,
\end{align}
where the solvency indicator $\xi = 1$ if $a > (1 - w^d) d$ and
$\xi = 0$ otherwise. Figure \ref{fig:sol_one} illustrates the solution
for different combinations of $w^s$ and $w^d$. In all cases $d = 1$
is fixed without loss of generality.
\begin{figure}[h]
  \centering
  \vspace*{-1cm}
  \includegraphics[width=0.9\textwidth]{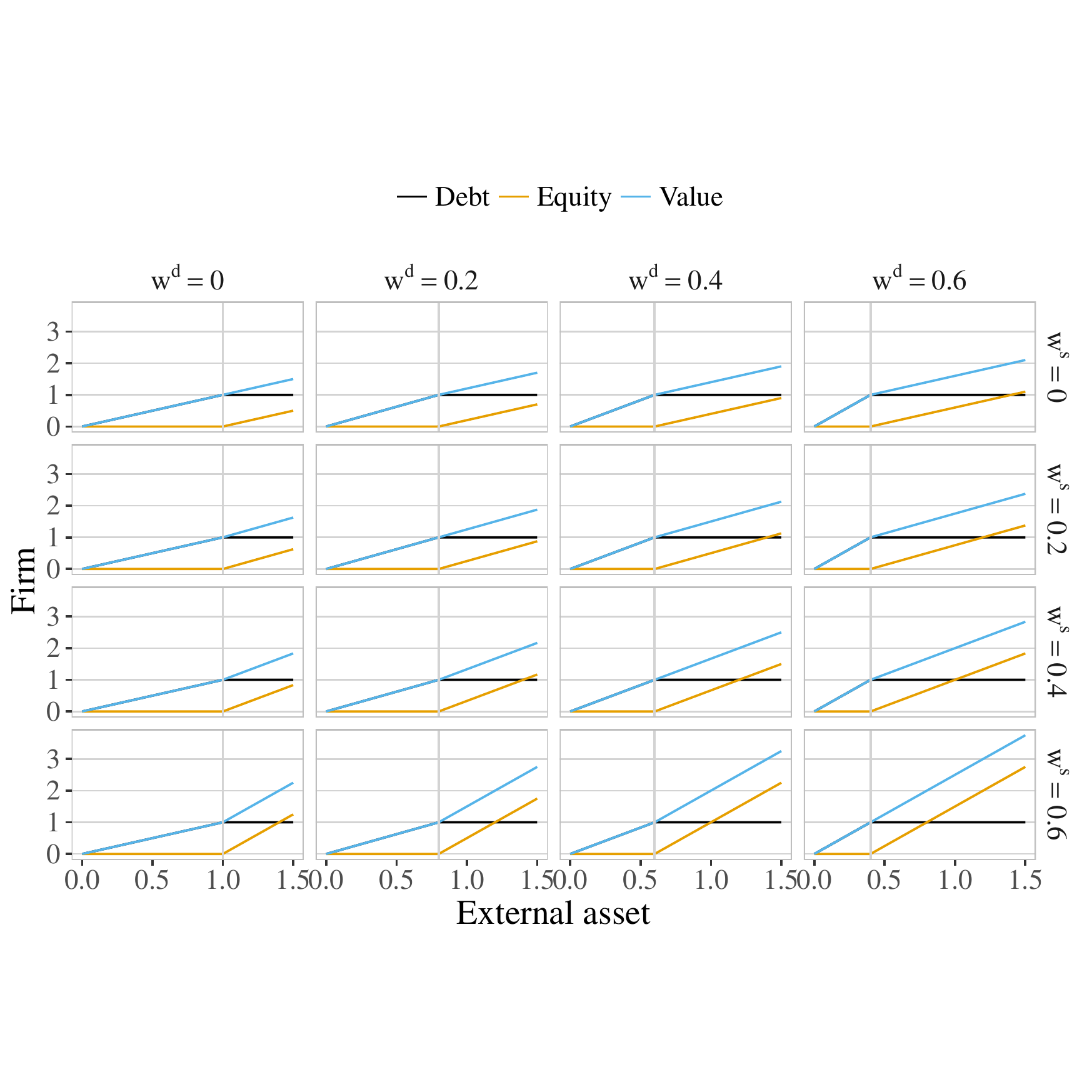} \\
  \vspace*{-1cm}
  \caption{Solution of \eq{sol_one} for cross-holding fractions of
    $w^s, w^d \in \{0, 0.2, 0.4, 0.6\}$.}
  \label{fig:sol_one}
\end{figure}
Notice how the value of the firms is increased by cross-holdings with a
clear difference between equity and debt holdings. Whereas debt
cross-holdings shift the default boundary to lower values and thus
increase the value of the firms when these are insolvent, equity
cross-holdings are ineffective in this regime and instead lead to
elevated values when firms are solvent. Furthermore, they do not
provide a risk-sharing benefit leaving the default boundary
unchanged. As shown below, this is an artefact of the symmetric
solution with a single external asset considered here and not true in
general.

Assuming a single external asset $A_t$ following a geometric Brownian
motion, we can compute the market values of debt and equity
analytically. As detailed in appendix \ref{app:sol_one}, we obtain
\begin{align}
  \label{eq:sol_one}
  s_t &= \frac{1}{1 - w^s} \left( a_t \Phi(d_{+}) - (1 - w^d) d e^{- r \tau} \Phi(d_{-}) \right) \\
  r_t &= \frac{1}{1 - w^d} \left( a_t \Phi(- d_{+}) + (1 - w^d) d e^{- r \tau} (1 - \Phi(- d_{-})) \right)
\end{align}
where $\Phi$ denotes the cumulative distribution function of a
standard normal and $d_{\pm}$ are defined as
\begin{align}
  \label{eq:sol_d_pm}
  d_{\pm} &= \frac{ \ln \frac{a_t}{(1 - w^d) d} + (r \pm \frac{1}{2} \sigma^2) \tau}{\sqrt{\tau} \sigma} \, .
\end{align}
As in the Merton model, equity can be considered a long call option
and the recovery value of debt is insured by a short put. Indeed, it
is easily checked that \eq{sol_one} can be written as
\begin{align}
  \label{eq:sol_BS}
  s_t &= \frac{1}{1 - w^s} C_{BS}(a_t, (1 - w^d) d, r, \tau, \sigma) \\
  r_t &= \frac{1}{1 - w^d} \left( e^{-r \tau} (1 - w^d) d - P_{BS}(a_t, (1 - w^d) d, r, \tau, \sigma) \right)
\end{align}
with $C_{BS}, P_{BS}$ denoting the Black-Scholes price of a call and
put option with spot price $a_t$ and strike $(1 - w^d) d$
respectively. Overall, two effects arise from cross-holdings. First,
the default boundary corresponding to the nominally required repayment
is lowered to $(1 - w^d) d$. Second, equity and debt values are
amplified by $\frac{1}{1 - w^s}$ and $\frac{1}{1 - w^d}$ respectively.
\begin{figure}[h]
  \centering
  \vspace*{-1cm}
  \includegraphics[width=0.9\textwidth]{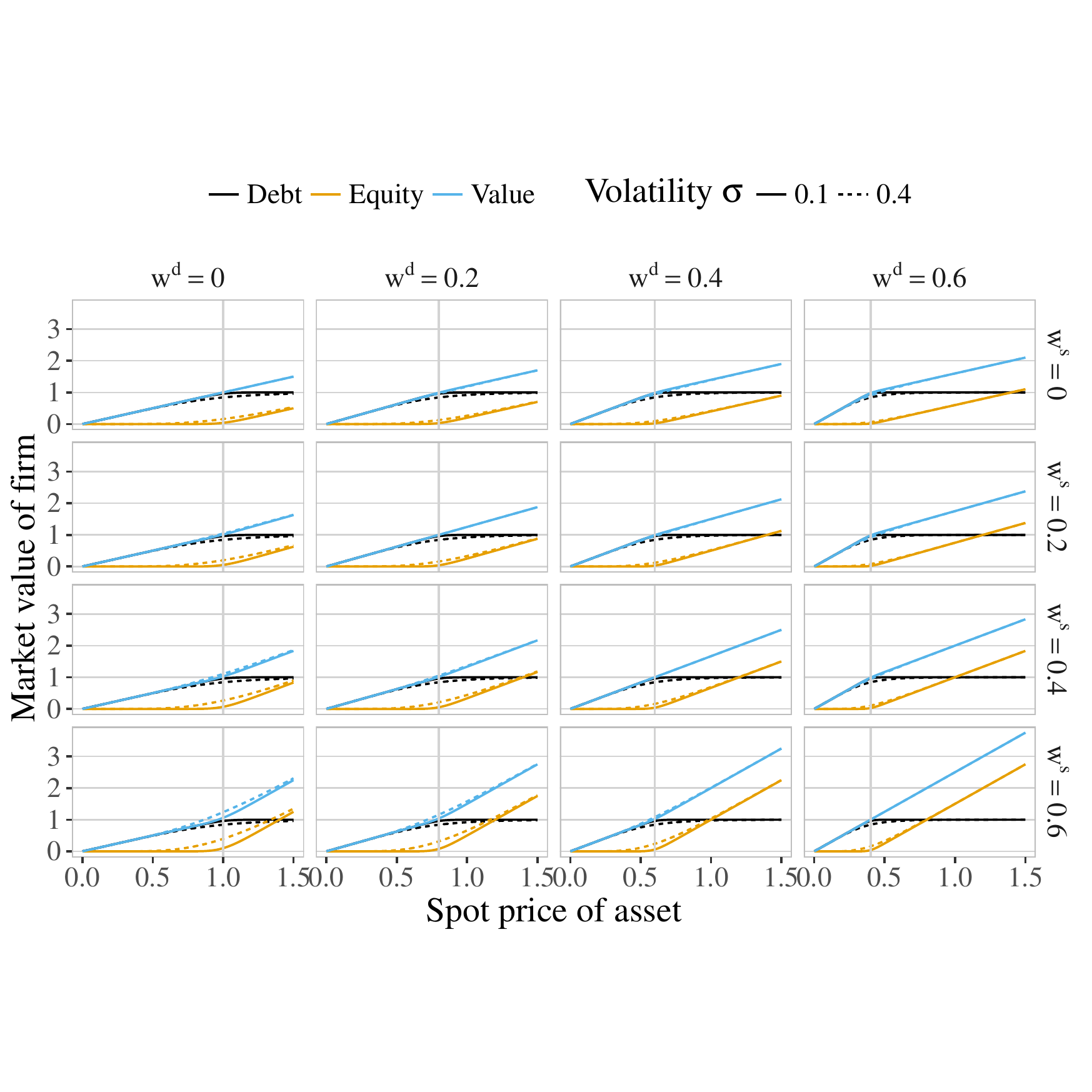} \\
  \vspace*{-1cm}
  \caption{Market values of equity and debt for cross-holding
    fractions of $w^s, w^d \in \{0, 0.2, 0.4, 0.6\}$ and volatilities $\sigma = 0.1, 0.4$.}
  \label{fig:mkt_one}
\end{figure}
Figure \ref{fig:mkt_one} shows the market values of equity and debt
together with the total value of the firm.  Compared to the values at
maturity, i.e. as in figure \ref{fig:sol_one}, the market values are
smoothed with equity values increased and debt values decreased
corresponding to the implicit long call and short put option
respectively. At higher volatilities this effect is even more
pronounced. Depending on the relative strength of equity to debt
cross-holdings, volatility thereby increases ($w^s > w^d$) or
decreases ($w^s < w^d$) the firms value.

To compute the Greeks, we can either resort to corollary
\ref{coro:fix_diff}, as illustrated in appendix \ref{app:sol_Greeks},
or simply use the corresponding Greeks from the Black-Scholes
formula. Indeed, \eq{sol_BS} implies that all Greeks are analytic
being amplifications of the standard Black-Scholes Greeks.
\begin{align}
  \frac{\partial s_t}{\partial a_t}
  &= \frac{1}{1 - w^s} \Phi(d_{+}) \\
  \frac{\partial r_t}{\partial a_t}
  &= \frac{1}{1 - w^d} \Phi(- d_{+}) \\
  \frac{\partial s_t}{\partial \sigma}
  &= \frac{1}{1 - w^s} a_t \varphi(d_{+}) \sqrt{\tau} \\
  \frac{\partial r_t}{\partial \sigma}
  &= - \frac{1}{1 - w^d} a_t \varphi(d_{+}) \sqrt{\tau} \\  
  - \frac{\partial s_t}{\partial \tau}
  &= - \frac{1}{1 - w^s} \left( \frac{a_t \varphi(d_{+}) \sigma}{2 \sqrt{\tau}}
  + r (1 - w^d) d e^{- r \tau} \Phi(d_{-}) \right) \\
  - \frac{\partial r_t}{\partial \tau}
  &= \frac{1}{1 - w^d} \left( r (1 - w^d) d e^{- r \tau}
    + \frac{a_t \varphi(d_{+}) \sigma}{2 \sqrt{\tau}}
    - r (1 - w^d) d e^{- r \tau} \Phi(- d_{-}) \right) \\
  \frac{\partial s_t}{\partial r}
  &= \frac{1}{1 - w^s} (1 - w^d) d \tau e^{- r \tau} \Phi(d_{-}) \\
  \frac{\partial r_t}{\partial r}
  &= - \frac{1}{1 - w^d} \left( (1 - w^d) d \tau e^{- r \tau}
    - (1 - w^d) d \tau e^{- r \tau} \Phi(- d_{-}) \right)
\end{align}
Figure \ref{fig:Greek_one} shows $\Delta, \mathcal{V}, \Theta$ and
$\rho$ with the parameters as above. Compared to no cross-holdings
$w^s = w^d = 0$ the amplification is clearly visible. Interestingly,
while risks are generally high in distress, i.e. when firms are close to their default boundary,
the risks with respect to $\sigma$, $\tau$ and $r$ vanish with
dropping asset values, i.e. outright default. $\Delta$ instead is
high and amplified by the cross-holdings of debt in this case.
Furthermore, $\Delta$ risk is always positive, in contrast to the other
risks which are hedged between equity and debt holders. Overall,
this suggests $\Delta$ as a suitable measure of the systemic risk arising from
cross-holdings. It is also interesting to compare $\Delta$ for the
value of the firm, i.e.
$\frac{\partial v}{\partial a_t} = \frac{1}{1 - w^s} \Phi(d_{+}) +
\frac{1}{1 - w^d} \Phi(- d_{+})$ with the systemic risk index
$\v{\pi}$ which is computed in appendix \ref{app:sol_Greeks} as
$\v{\pi} = \frac{1}{1 - w^s} \Phi(d_{-}) + \frac{1}{1 - w^d} \Phi(-
d_{-})$. Thus, the difference of the Black-Scholes model between
percent moneyness and $\Delta$ also appears in this setup.
\begin{figure}[h]
  \centering
  \vspace*{-1cm}
  \includegraphics[width=0.45\textwidth]{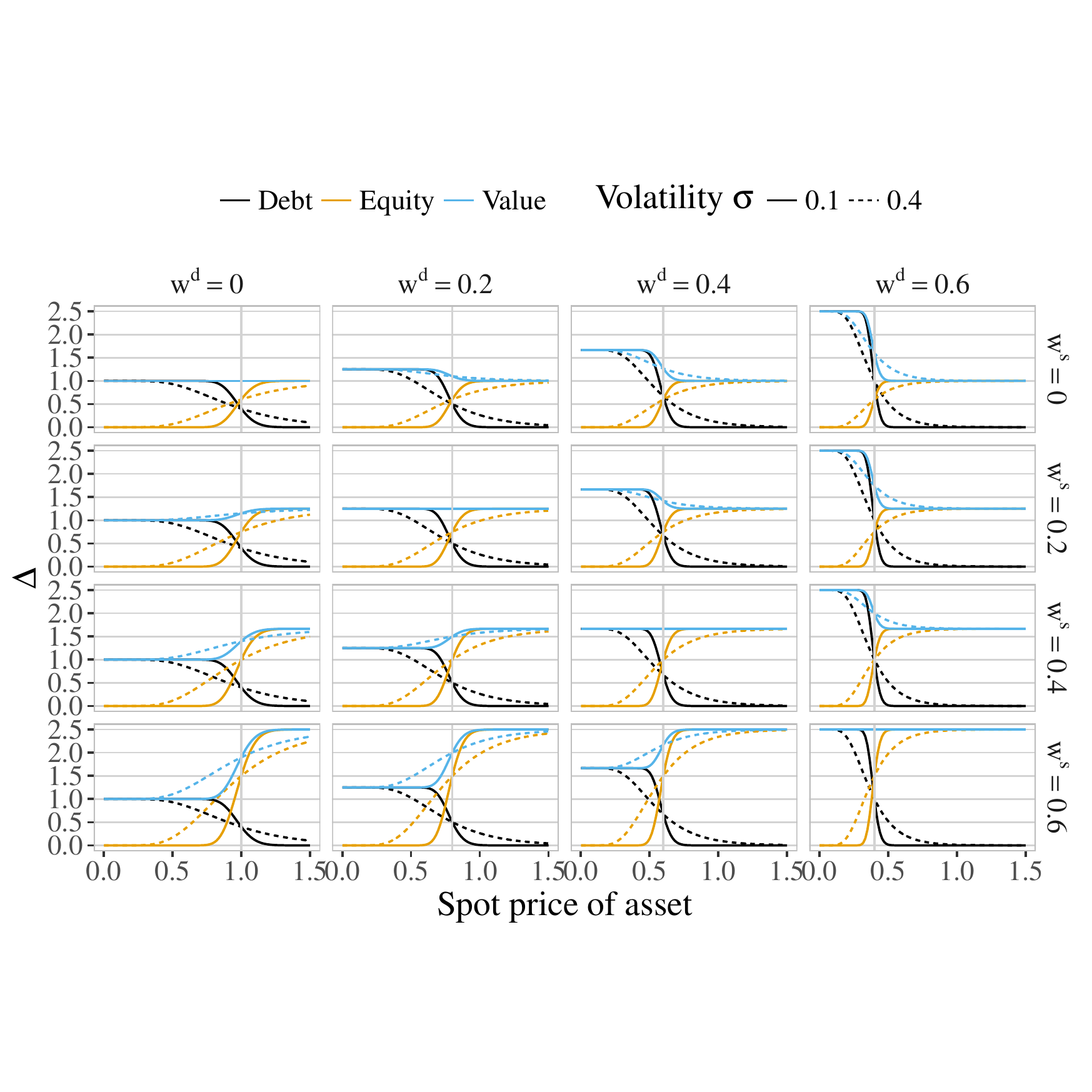}
  \includegraphics[width=0.45\textwidth]{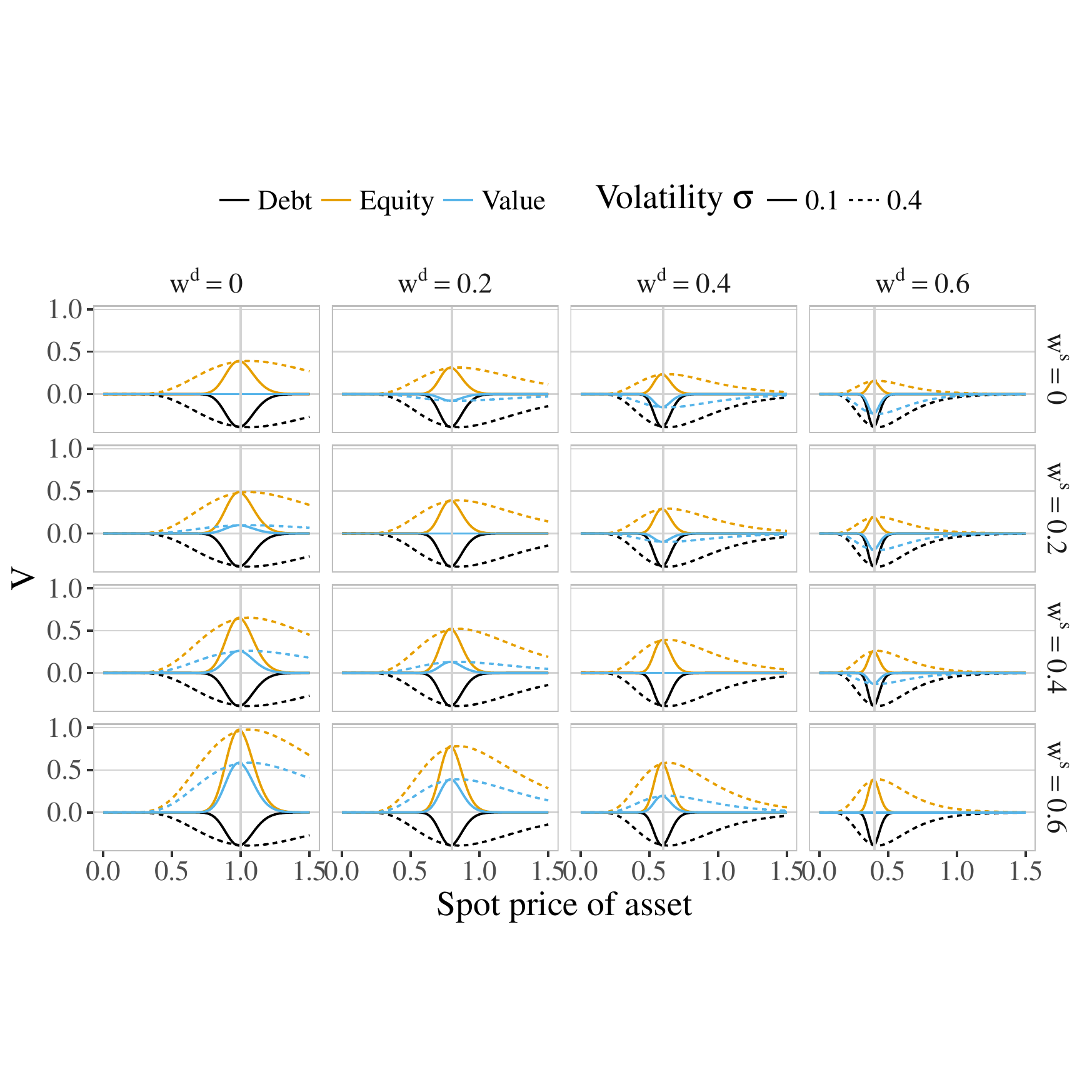} \\
  \vspace*{-2cm}
  \includegraphics[width=0.45\textwidth]{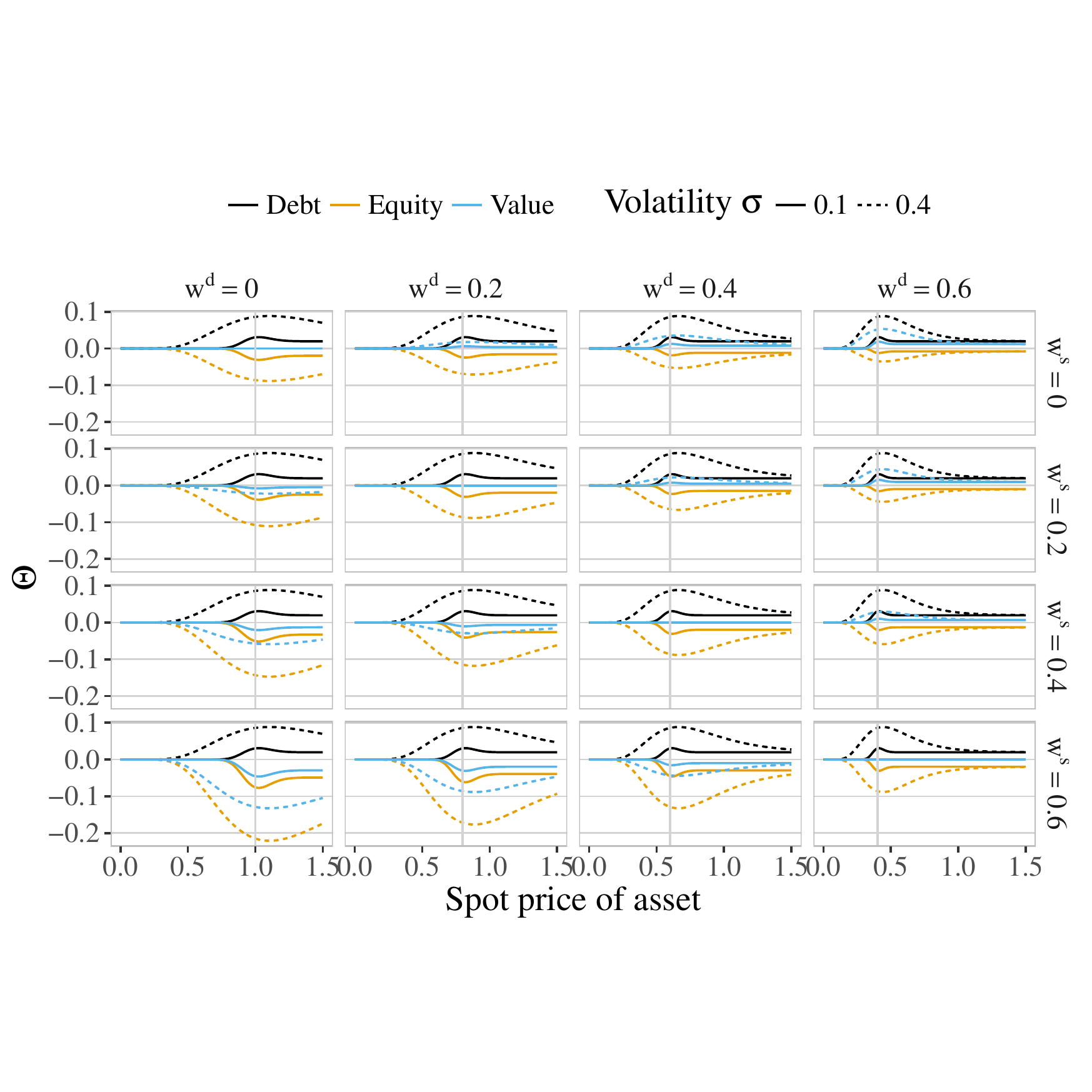}
  \includegraphics[width=0.45\textwidth]{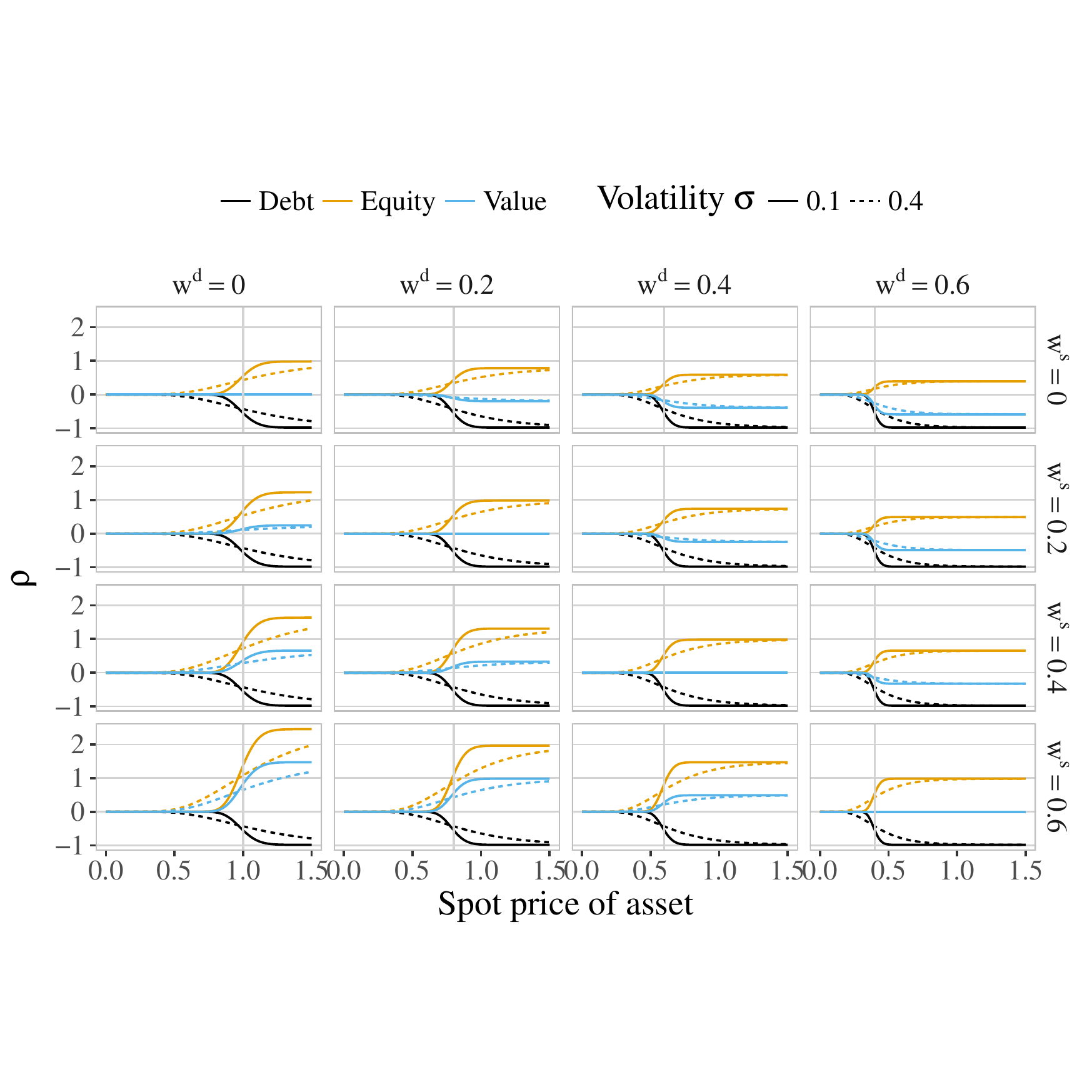} \\
  \vspace*{-1cm}
  \caption{Greeks $\Delta, \mathcal{V}, \rho$ and $\Theta$ for cross-holding
    fractions of $w^s, w^d \in \{0, 0.2, 0.4, 0.6\}$ and volatilities $\sigma = 0.1, 0.4$.}
  \label{fig:Greek_one}
\end{figure}

\clearpage

\subsection{Random networks example}

We consider again $n$ firms, each now holding a different
external asset $a_i$. This breaks the symmetry of the solution
considered above and is in general not analytically tractable. Indeed,
already in case of $n = 2$ firms the risk-neutral expectations in
\eq{val_Q} are intractable, even though the fixed point
$\v{x}^*(\v{A}_T)$ can be solved analytically
\citep{Suzuki2002,Karl2015}. Figure \ref{fig:two_firm} illustrates the
joint distribution of firm values arising from independent
log-normally distributed external assets. Parameters are choosen as in
figure 6 of \cite{Karl2015}, i.e. $r = 0, \tau = 1, a_0 = 1, \sigma = 1$
for the log normal distribution of $A_T$ and
$w_{12}^d = w_{21}^d = 0.95, w_{12}^s = w_{21}^s = 0, d_1 = d_2 =
11.3$
for the cross-holdings and nominal debt of both firms. Compared to the
value of external assets, firm values are increased considerably by
the debt cross-holdings and strongly distorted. Especially joint
defaults are non independent and accompanied by strongly correlated
firm values.
\begin{figure}[h]
  \centering
  \vspace*{-1cm}
  \includegraphics[width=0.9\textwidth]{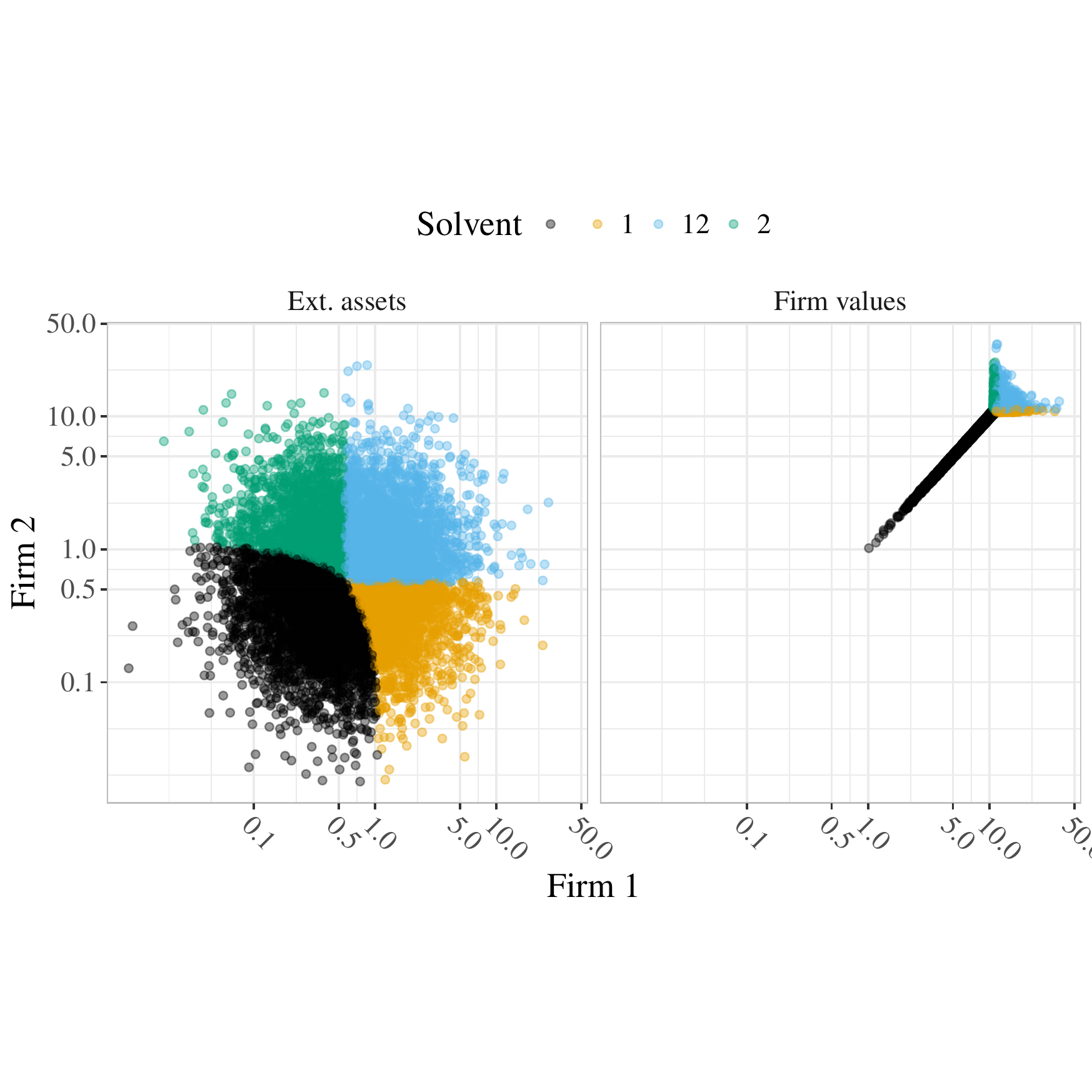}
  \vspace*{-1cm}
  \caption{External asset and firm values for two firms with
    parameters $r = 0, \tau = 1, a_0 = 1, \sigma = 1$ for the log normal
    distribution of $A_T$ and
    $w_{12}^d = w_{21}^d = 0.95, w_{12}^s = w_{21}^s = 0, d_1 = d_2 =
    11.3$ as in figure 6 of \cite{Karl2015}.}
  \label{fig:two_firm}
\end{figure}

Thus, in this section we resort to numerical methods in order to
compute network Greeks. In particular, we use Monte-Carlo integration
to approximate the risk neutral expectation of \eq{greek_net}.
Furthermore, as in \cite{Gai2010} we consider
debt cross-holdings only. The network of cross-holdings is generated
at random according to the Erd\H{o}s-R{\'e}nyi model, i.e. each firm
is connected to each counterparty with a fixed probability $p$. The
number of incoming $k^{\text{in}}$ and outgoing connections $k^{\text{out}}$ then
follow Poisson distributions, both with a mean
$\langle k \rangle = n p$. The Erd\H{o}s-R{\'e}nyi model exhibits a
phase transition at $\langle k \rangle = 1$ where a giant component,
i.e. connected subgraph infinite size in the thermodynamic limit
$n \to \infty$, of connected firms appears. Here, we adjust the
connection probabilities such that the average number of connected
counterparties covers this transition and varies between $0$ and
$5$. The actual weights of cross-holdings, i.e. $\v{M}^d$, are then
obtained by scaling the random adjacency matrices such that
$\sum_i M_{ij}^d = w^d \, \forall j = 1, \ldots, n$ whenever $j$ has
any outgoing connection. Otherwise, $\sum_i M_{ij}^d = k_j^{\text{out}} = 0$.
To illustrate the effect of different strength of debt cross-holdings
$w_d$ is varied between $0$ and $0.6$. Note that in the simulation
studies of \cite{Gai2010} the incoming connections, corresponding to
the investment portfolio of each firm, are scaled such that the same
total amount is held with each counterparty. In our case, this would
correspond to requiring that $M_{ij}^d = \frac{w^d}{k_i^{\text{in}}}$ which,
in general, cannot be ensured together with the above constraint on
the outgoing connections. The Sinkhorn-Knopp algorithm \citep{Sinkhorn1967,Idel2016}
allows to achieve fixed row and column sums simultaneously by
iteratively rescaling rows and columns. On matrices with some entries
exactly zero, as in our case of firms not holding debt from every
possible counterparty, such a rescaling is not always possible and the
algorithm does not necessarily converge. Thus, in simulations a rejection step is
required changing the support of the random network example. Yet, as
the results were almost unchanged, all simulations in this section are
based on the simpler version of scaling the outgoing connections only,
i.e. $\sum_i M_{ij}^d = w^d \; \forall j$. Note that this does not leave the incoming
connections unconstrained but fixes their average weight at $w^d$ as well as
$\frac{1}{n} \sum_j (\sum_i M_{ij}^d) = w^d = \frac{1}{n} \sum_i (\sum_j M_{ij}^d)$.

Each firm is assumed to hold a different external asset, whose values
at maturity are independent and log normally distributed. We fix the
risk neutral interest rate $r = 0$, time to maturity $\tau = 1$ and
volatility $\sigma = 0.4$, varying the initial asset prices between
$a_0 = 0.1$ and $2.5$. Note that by \eq{val_Q} this fixes the market
values of debt and equity of each firm. Thus, in contrast to
\cite{Gai2010}, we cannot fix the fraction of external assets
$\frac{a}{v}$ and capital ratios $\frac{s}{v}$ independently. Instead
they are derived from the market prices following from the chosen
parameters. Figure \ref{fig:sol_ER_market} shows the resulting values
together with the market prices of equity and debt as well as the
default probability. For each combination of average connectivity
$\langle k \rangle$, external asset value $a_0$ and strength of debt
cross-holdings $w^d$, we simulated 1000 networks of $n = 60$
firms. This size was chosen as it showed small finite size
effects, yet is small enough to efficiently compute network valuations
and partial derivatives. For each network, 700 normal random vector
$\v{Z}$ where drawn and used to compute market prices (\eq{val_Q}) and
network Greeks (\eq{greek_Q}). The fixed point
$\v{x}^*(\v{a}_T(\v{Z}))$ was found by Picard iteration of the map
$\v{g}$ defined in \eq{XOS_g} which was shown to work efficiently in
this model \citep{Hain2015}. Figures show the mean values over all random
networks, asset price draws. As on average all firms are symmetric
under the scaled Erd\H{o}s-R{\'e}nyi model, we also averaged over
firms, i.e. showing values for a typical firm.

\begin{figure}[h]
  \centering
  \includegraphics[width=0.9\textwidth]{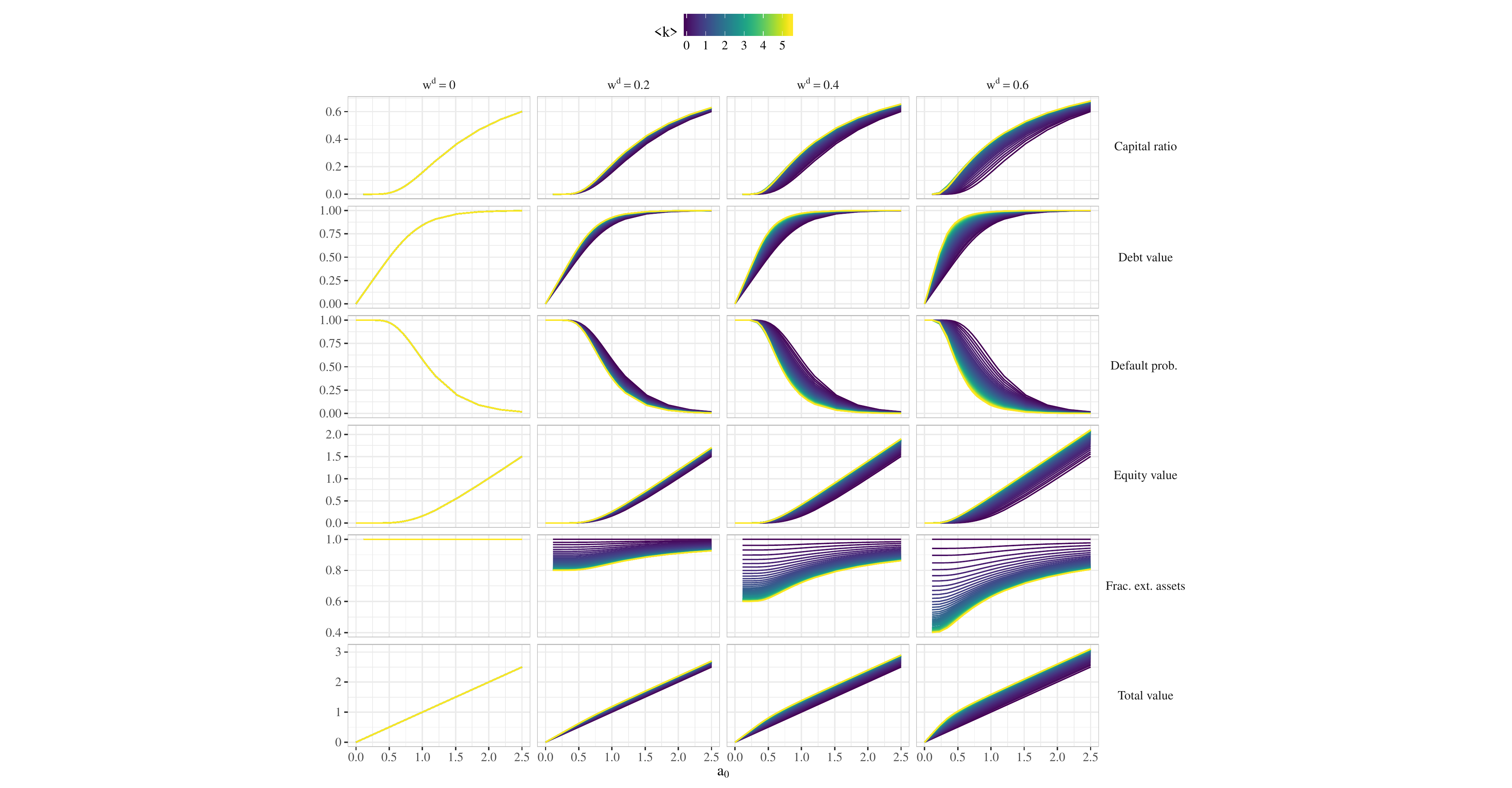}
  \caption{Market values of Erd\H{o}s-R{\'e}nyi random debt cross-holding
    networks of $60$ firms. At all strength $w^d$ of cross-holdings
    equity, debt and total value of firms increase
    with  diversification, i.e. average connectivity
    $\langle k \rangle$.
    For comparison with \cite{Gai2010} also
    the capital ratio, default probability and fraction of external
    assets are shown. Note that these are derived from market prices
    and thus endogenous in our model. Parameters are
    $r = 0, \tau = 1, \sigma = 0.4, \v{M}^s = \v{0}$ and
    $a_0, w_d, \langle k \rangle$ varying as denoted in the figure.}
  \label{fig:sol_ER_market}
\end{figure}
Figure \ref{fig:sol_ER_market} clearly illustrates the beneficial
effect of debt cross-holdings. For a fixed spot value of external
assets $a_0$, the market prices of equity and debt increase with the
number $\langle k \rangle$ and strength $w^d$ of
connections. Correspondingly, this diversification benefit leads to
reduced default probabilities. Despite the wide set of parameters
considered, few combinations lead to values of capital ratios of
$\approx 4\%$, default probabilities of a few percent and fractions of
external assets of $\approx 80\%$ which were considered realistic in
\cite{Gai2010}. At present, we are not sure if this is a fundamental
limitation of the model, i.e. missing important, qualitative features
of real markets, or could be remedied by more carefully chosen
parameters.

\begin{figure}[h]
  \centering
  \includegraphics[width=0.48\textwidth]{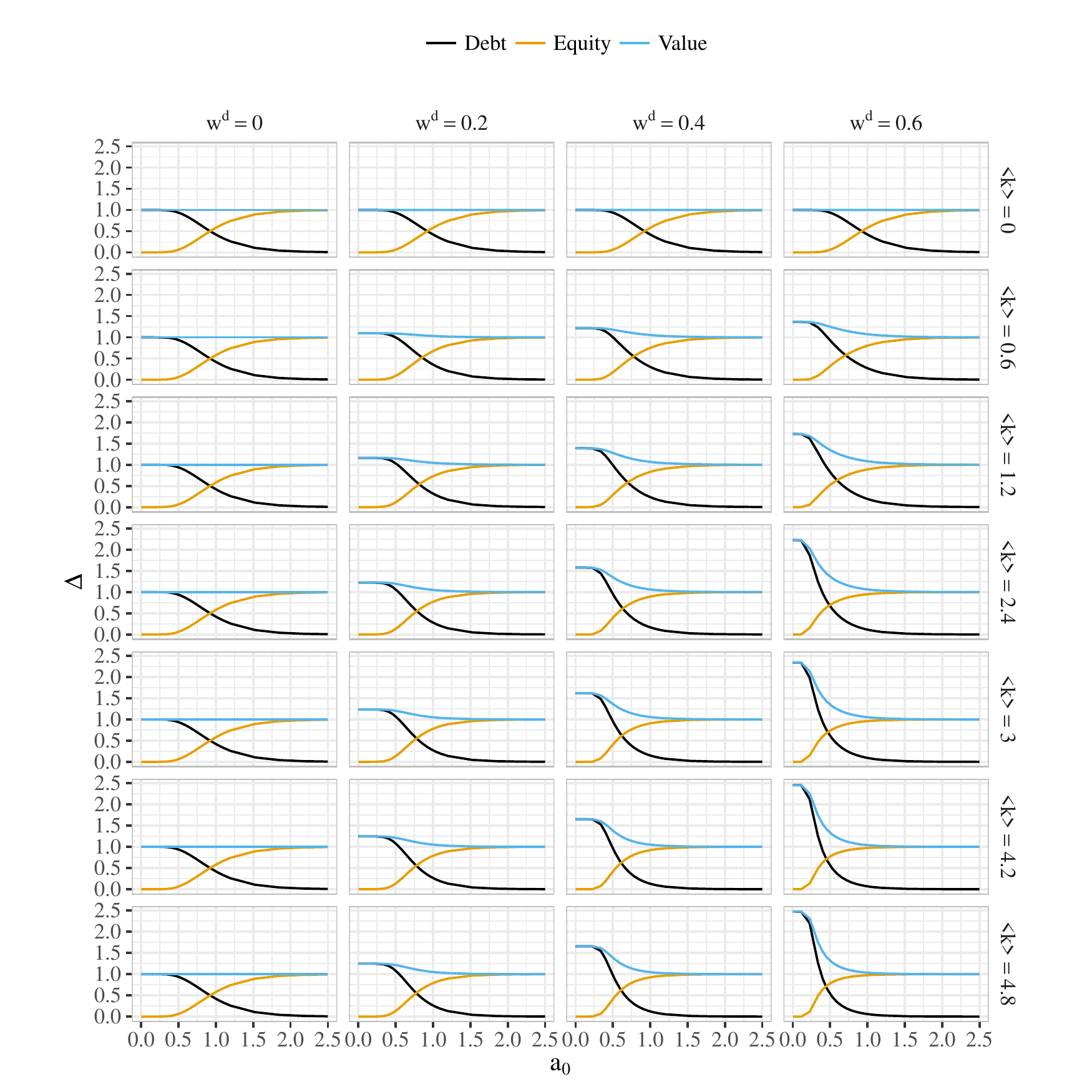}
  \includegraphics[width=0.48\textwidth]{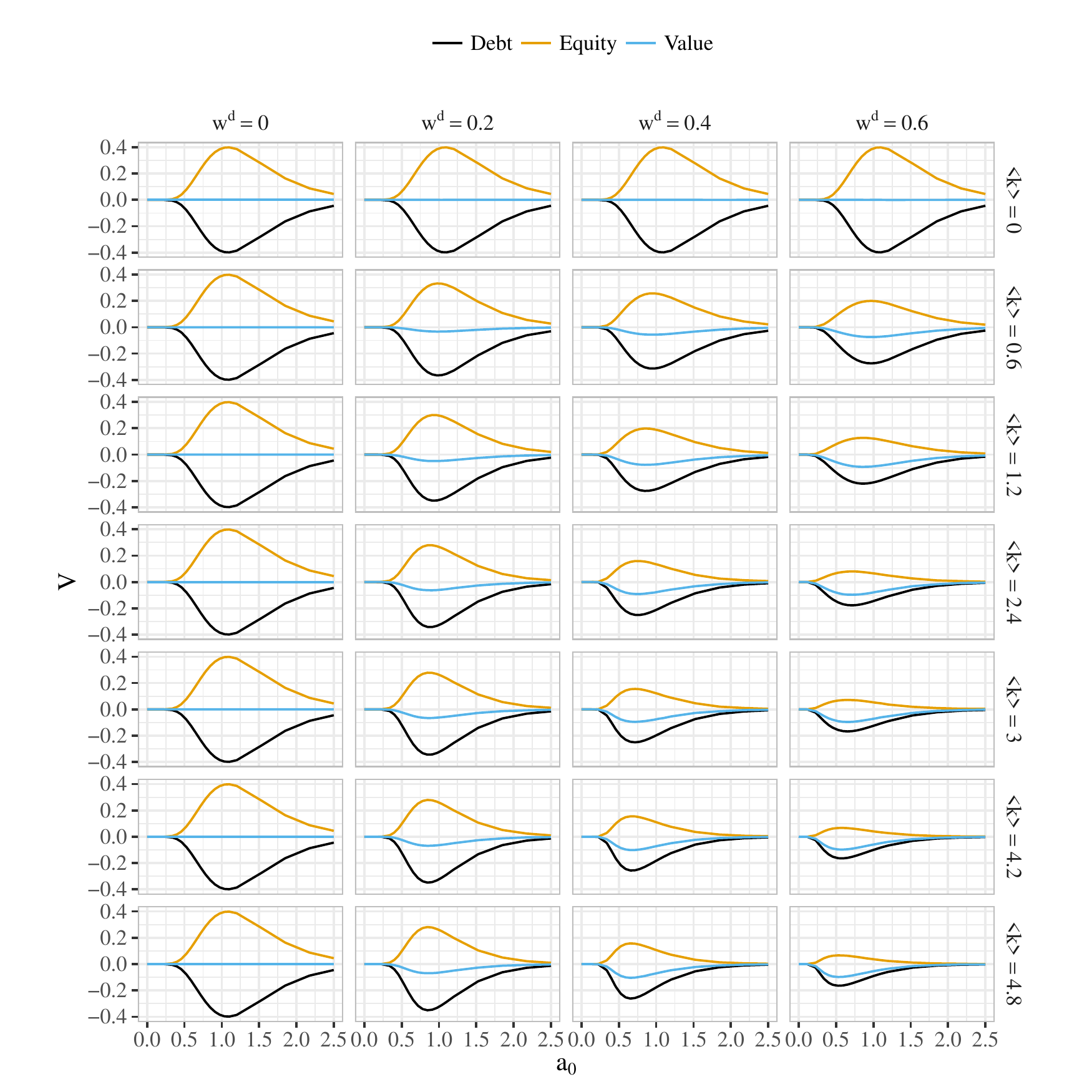} \\
  \includegraphics[width=0.48\textwidth]{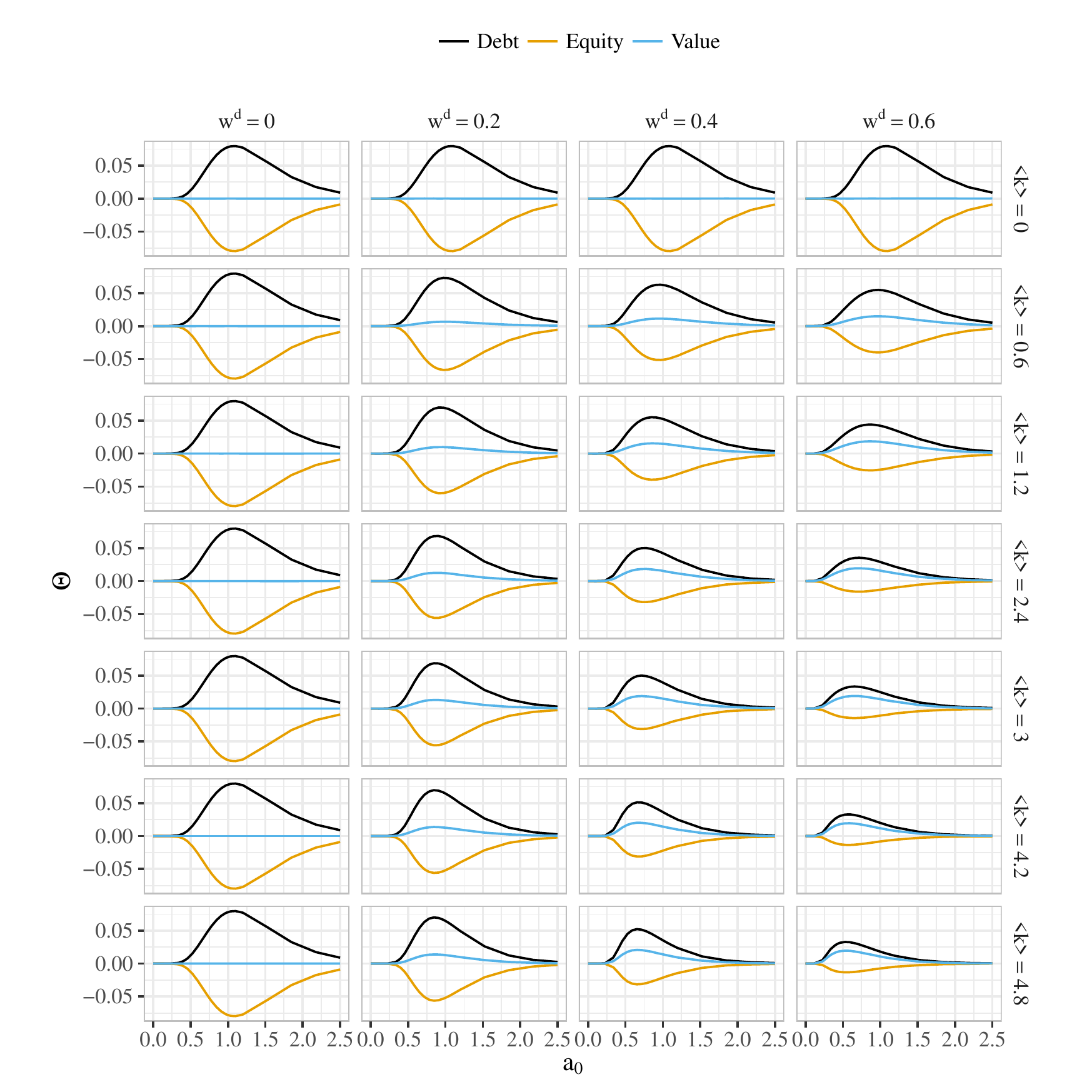}
  \includegraphics[width=0.48\textwidth]{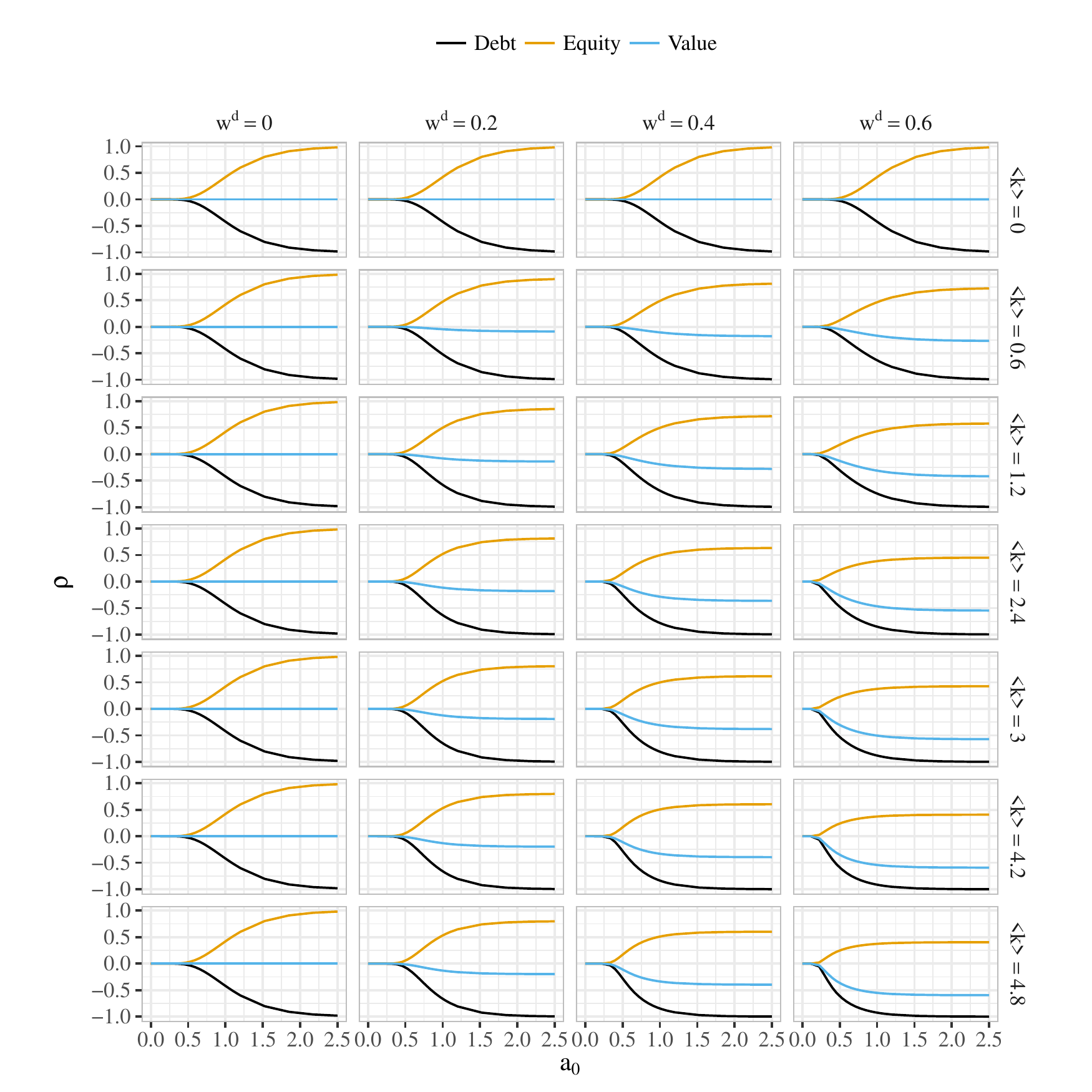}
  \caption{Greeks $\Delta, \mathcal{V}, \Theta$ and $\rho$ for
    Erd\H{o}s-R{\'e}nyi random cross-holding networks depending on the
    spot price $a_0$ of external assets. In all cases the average
    impact of each firm on the equity, debt and value of all other firms is
    shown. Parameters were chosen as in figure
    \ref{fig:sol_ER_market}.}
  \label{fig:sol_ER_Greek}
\end{figure}
Next, in order to study the impact of network parameters on systemic
risk, we compute the network Greeks. Fig.~\ref{fig:sol_ER_Greek}
shows the results for some
selected connectivities and should be compared to figure
\ref{fig:Greek_one}. Again, as firms are on average symmetric in our
model, we show the average total impact on all firms, e.g.
$\widehat{\Delta}^{\text{Total}} = \v{1}^T \v{\Delta} \frac{1}{n} \v{1}$
which is denoted as ``Value $\Delta$'' in the figure. As in the
case of the strictly symmetric solution, debt and equity react
differently to risk factors. $\widehat{\Delta}^{\text{Total}}$ is strongly amplified
at low external asset values by debt cross-holdings, increasing with
both the number and strength of connections. With the chosen
parameters the amplification factor is bounded by $\frac{1}{1 - w^d}$
due to the employed scaling of cross-holding positions.
Overall, the Greeks exhibit a similar behavior as in the strictly symmetric case.
Interestingly, the new parameter of average connectivity $\langle k \rangle$
tends to decrease the sensitivities to most
risk factors except for $\Delta$ which is strongly amplified.
The dampening effect is stronger on equity though such that the firms values
are no longer $\mathcal{V}, \Theta$ or $\rho$-neutral.

\begin{figure}[h]
  \centering
  \vspace*{-1cm}
  \includegraphics[width=0.95\textwidth]{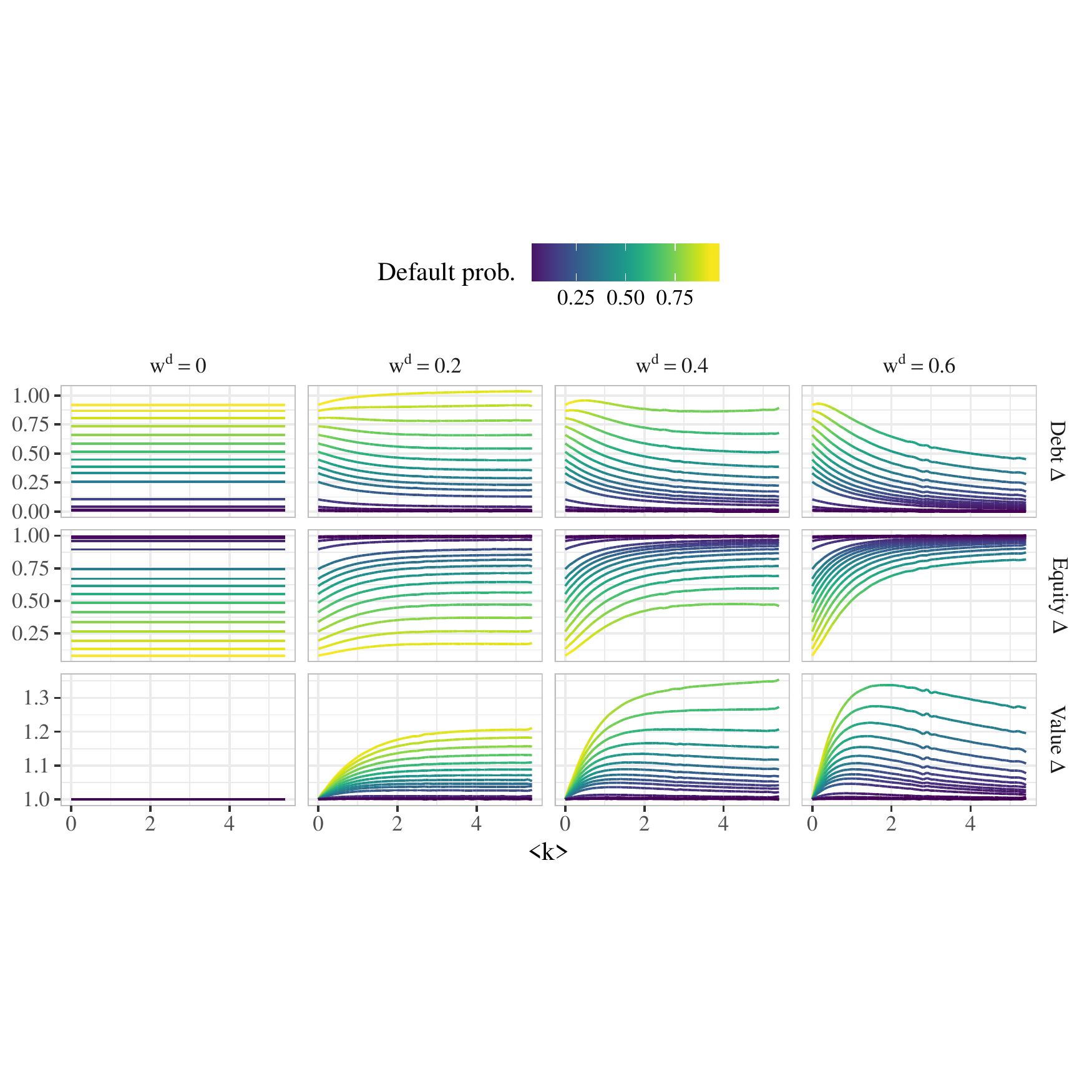}
  \vspace*{-1cm}
  \caption{$\widehat{\Delta}^{\text{Total}}$ depending on the average
    number $\langle k \rangle$ of counterparties. Especially with
    large cross-holdings of debt $w_d$ the impact on the firm value
    exhibits a contagion window. This arises from the interplay
    between the amplification due to higher connectivity and the
    diversification benefit lowering the default probability.
    To illustrate this effect, the different
    lines -- corresponding to different asset spot prices $a_0$ -- are
    colored by default probability.
  }
  \label{fig:sol_ER_window}
\end{figure}
Now, we focus on $\Delta$ which is readily interpreted as the
(first-order) impact of asset price shocks. Thus, it should be
comparable to contagion arising from defaults as in the model of
\cite{Gai2010}. In this model, the appearance of a contagion window
was noted: Contagion starts to spread at the phase transition of the
random network, i.e. when firms are sufficiently connected such that
large connected clusters appear, and stops when the impact from single
counterparties becomes too small, i.e. when firms are sufficiently
diversified. Here, we obtain similar results as shown in figure
\ref{fig:sol_ER_window}. Especially at larger strength of debt
cross-holdings $w^d$, we observe that the total impact of an asset
price shock on all firms (Value
$\Delta = \widehat{\Delta}^{\text{Total}}$) increases and then
decreases with increasing connectivity. There is no clearly defined
contagion window though. Instead, the phase transition of the
underlying random network is invisible with contagion -- as quantified
by $\widehat{\Delta}^{\text{Total}}$ -- starting to rise already at
$\langle k \rangle > 0$. Furthermore, diversification can reduce
contagion but never prevent it completely. Overall, the appearance of
a contagion maximum, akin to a window, is driven by an intricate
interplay in the amplification driven by \eq{net_inv}: While
increasing connectivity in $\v{M}^d$ would lead to higher
amplification, at the same time the default probability drops and
reduces amplification. Interestingly, both effects balance each other
in a way that the critical connectivity of the phase transition is
invisible, i.e. initially any connection quickly increases
contagion. On the other end, even at full connectivity contagion does
not vanish as market prices react to the potential contagion that
would arise if firms default at maturity. I.e. as long as the default
probability is non-zero, market prices reflect the potential credit
risk and thus react to asset price shocks. This is also the reason
that the diversification benefit is limited by the extend to which the
probability of default is reduced. Further studies are certainly
necessary to investigate and understand the properties of $\Delta$ and
the other Greeks when applied in a network context.

\clearpage

\section{Discussion}
\label{sec:discuss}

We have investigated network valuation models which extend structural
credit risk models to a multi-firm setup. Taking the resulting
connection with derivative pricing seriously, we have shown how to
compute network Greeks. Our solution is analytic to the point where
risk-neutral expectations taking ex-post values at maturity to ex-ante
market prices have to be evaluated. Nevertheless, we have computed the
network Greeks in large random cross-holding networks via Monte-Carlo
sampling. We believe that our model has great potential and many
merits for investigating and understanding systemic risk. First and
foremost it is firmly routed in established theories of asset pricing,
providing a sound and principled approach for systemic risk
analysis. From this perspective, we advocate the view that the
valuation of financial contracts and systemic risk should not be
studied as separate subjects, which has also been expressed by
Fischer. Furthermore, the following interesting observations have been
derived from the model: The effectiveness of contagion changes based
on the distance to default of firms. Especially debt cross-holdings
are mostly invisible when firms are running strong, yet amplify asset
price shocks when firms become distressed. Thus, effective management
of systemic risk cannot be based on current market prices alone but
needs to include crisis scenarios. Second, network Greeks provide a
principled way for quantifying risk in financial networks, connecting
systemic risk research with well established risk management
practices. In particular, we clarified the relation between our
network $\Delta$ and the threat index proposed by
\cite{Demange2018}. We believe that $\Delta$ is preferable as it
quantifies the impact of asset price shocks at market, i.e. ex-ante,
values and considers the impact on the total firms value, instead of
debt repayments alone. Third, in the considered model, risk is
redistributed differently between firms within the network and
external investors funding them. From their outside perspective no
amplification takes place with risk merely redistributed without much
sharing benefits between them. Finally, our framework is very general
and applies almost unchanged to extensions of the considered model. By
including equity cross-holdings the model is already more general than
most studies of systemic risk focusing on default contagion,
i.e. arising from debt cross-holdings. Several interesting questions
suggest themselves for future research, e.g. cross-holding of
contracts with different seniorities \citep{Fischer2014} or roll-over
and liquidity risk as studied by \cite{He2012} in the case of single
firms. While some of these extensions appear immediate, the inclusion
of dead-weight losses at default \cite{Battiston2012} poses some
challenges as the valuation becomes non-continuous.

\subsection*{Acknowledgement}

Nils Bertschinger thanks Dr. h.c. Maucher for funding his position.

\appendix

\section{Proof of hypothesis \ref{hyp:mono}}
\label{app:hyp}

\begin{proof}
From corollary \ref{coro:fix_diff} and \eq{g_diff_mat} we obtain
\begin{align}
  \left[ \begin{array}{c} \v{u}^s \\ \v{u}^d \end{array} \right] 
&= \left[ \v{I}_{2 n \times 2 n} - \frac{\partial}{\partial \v{x}} \v{g}(\v{a}, \v{x}) \right]^{-1}
  \left[ \begin{array}{c} \diag(\v{\xi}) \\ \diag(\v{1}_n - \v{\xi}) \end{array} \right] \\[1.5ex]
\Leftrightarrow \left[ \v{I}_{2 n \times 2 n} - \frac{\partial}{\partial \v{x}} \v{g}(\v{a}, \v{x}) \right]
  \left[ \begin{array}{c} \v{u}^s \\ \v{u}^d \end{array} \right] 
&= \left[ \begin{array}{c} \diag(\v{\xi}) \\ \diag(\v{1}_n - \v{\xi}) \end{array} \right] \\[1.5ex]
  \label{eq:app_hyp:fix}
  \left[ \begin{array}{c} \v{u}^s \\ \v{u}^d \end{array} \right] 
&= \left[ \begin{array}{c} \diag(\v{\xi}) \\ \diag(\v{1}_n - \v{\xi}) \end{array} \right]
  + \left[
    \begin{array}{cc}
      \diag(\v{\xi}) \v{M}^s & \diag(\v{\xi}) \v{M}^d \\
      \diag(\v{1}_n - \v{\xi}) \v{M}^s & \diag(\v{1}_n - \v{\xi}) \v{M}^d
    \end{array}
    \right] \left[ \begin{array}{c} \v{u}^s \\ \v{u}^d \end{array} \right] \, .
\end{align}
Thus, the partial drivatives are solutions of the fixed point
$\left[ \begin{array}{c} \v{u}^s \\ \v{u}^d \end{array} \right] =
T_{\v{\xi}} \left[ \begin{array}{c} \v{u}^s \\ \v{u}^d \end{array}
\right]$
with $T_{\v{\xi}}$ denoting the map on the right hand side of
\eq{app_hyp:fix}. We now show that (i) $T_{\v{\xi}}$ is monotonically
increasing in both $\v{u}^s$ and $\v{u}^d$, and (ii) $T_{\v{\xi}}$ is
a contraction.
\begin{enumerate}[(i)]
\item Consider
  $\left[ \begin{array}{c} \v{u'}^s \\ \v{u'}^d \end{array} \right]
  \geq \left[ \begin{array}{c} \v{u}^s \\ \v{u}^d \end{array}
  \right]$. Then
  \begin{align}
    T_{\v{\xi}} \left( \left[ \begin{array}{c} \v{u'}^s \\ \v{u'}^d \end{array} \right]
    - \left[ \begin{array}{c} \v{u}^s \\ \v{u}^d \end{array} \right] \right)
    &= \left[
    \begin{array}{cc}
      \diag(\v{\xi}) \v{M}^s (\v{u'}^s - \v{u}^s) & \diag(\v{\xi}) \v{M}^d (\v{u'}^d - \v{u}^d) \\
      \diag(\v{1}_n - \v{\xi}) \v{M}^s (\v{u'}^s - \v{u}^s) & \diag(\v{1}_n - \v{\xi}) \v{M}^d (\v{u'}^d - \v{u}^d)
    \end{array}
    \right] \geq \left[ \begin{array}{c} \v{0} \\ \v{0} \end{array} \right]
  \end{align}
  where the last inequality follows from
  $\v{u'}^{s,d} - \v{u}^{s,d} \geq \v{0}, \v{\xi} \in \mathbb{R}^n$
  and assumption \ref{assu_1} stating that $M^{s,d} \geq \v{0}$.
\item To show that $T_{\v{\xi}}$ is a contraction, we compute
  \begin{align}
    \norm{T_{\v{\xi}} \left[ \begin{array}{c} \v{u'}^s \\ \v{u'}^d \end{array} \right]
    - T_{\v{\xi}} \left[ \begin{array}{c} \v{u}^s \\ \v{u}^d \end{array} \right]}
    &= \norm{ \left[
    \begin{array}{cc}
      \diag(\v{\xi}) \v{M}^s & \diag(\v{\xi}) \v{M}^d \\
      \diag(\v{1}_n - \v{\xi}) \v{M}^s & \diag(\v{1}_n - \v{\xi}) \v{M}^d
    \end{array} \right]
                                         \left( \left[ \begin{array}{c} \v{u'}^s \\ \v{u'}^d \end{array} \right] 
    - \left[ \begin{array}{c} \v{u}^s \\ \v{u}^d \end{array} \right] \right)
    } \\[1.5ex]
    &\leq \norm{\left[
    \begin{array}{cc}
      \diag(\v{\xi}) \v{M}^s & \diag(\v{\xi}) \v{M}^d \\
      \diag(\v{1}_n - \v{\xi}) \v{M}^s & \diag(\v{1}_n - \v{\xi}) \v{M}^d
    \end{array} \right]} \norm{\left[ \begin{array}{c} \v{u'}^s \\ \v{u'}^d \end{array} \right] 
    - \left[ \begin{array}{c} \v{u}^s \\ \v{u}^d \end{array} \right]} \\[1.5ex]
    &\leq \lambda \norm{\left[ \begin{array}{c} \v{u'}^s \\ \v{u'}^d \end{array} \right] 
    - \left[ \begin{array}{c} \v{u}^s \\ \v{u}^d \end{array} \right]}
  \end{align}
  where $\lambda < 1$ exists when assuming the slightly stronger
  requirement that $\sum_i M^{s,d}_{i,j} < 1$ for all
  $j = 1, \ldots, n$. The contraction then follows as all columns of $\left[
    \begin{array}{cc}
      \diag(\v{\xi}) \v{M}^s & \diag(\v{\xi}) \v{M}^d \\
      \diag(\v{1}_n - \v{\xi}) \v{M}^s & \diag(\v{1}_n - \v{\xi}) \v{M}^d
    \end{array} \right]$ sum to less than one bounding its $\norm{\cdot}$ below one as well.
\end{enumerate}
By (ii) we know from Banach's fixed point theorem that the iteration
$\v{u}_{n+1} = T_{\v{\xi}} \v{u}_n$ converges to the unique fixed
point $\v{u}^*$ from any initial condition $\v{u}_0$. Further, from
(i) the convergence is strictly from above or below in both parts
$\v{u}^{s}$ and $\v{u}^d$ when $\v{u}^{s,d} \geq (\v{u}^{s,d})^*$ or
$\v{u}^{s,d} \leq (\v{u}^{s,d})^*$ respectively. Now denote the solution with $\v{\xi}$ with
$\v{u}_{\v{\xi}}$, i.e. $\v{u}_{\v{\xi}} = T_{\v{\xi}} \v{u}_{\v{\xi}}$, and consider
$\v{\xi'} \geq \v{\xi}$, i.e. with more firms being solvent. Then,
\begin{align}
  T_{\v{\xi'}} \left[ \begin{array}{c} \v{u}^s_{\v{\xi}} \\ \v{u}^d_{\v{\xi}} \end{array} \right]
  &= \left[
    \begin{array}{c}
      \diag(\v{\xi'}) + \diag(\v{\xi'}) \v{M}^s \v{u}^s_{\v{\xi}} + \diag(\v{\xi'}) \v{M}^d \v{u}^d_{\v{\xi}} \\
      \diag(\v{1}_n - \v{\xi'}) + \diag(\v{1}_n - \v{\xi'}) \v{M}^s \v{u}^s_{\v{\xi}} + \diag(\v{1}_n - \v{\xi'}) \v{M}^d \v{u}^d_{\v{\xi}}
    \end{array}
    \right] \\[1.5ex]
  &{\begin{array}{c} \geq \\ \leq \end{array}}
  \left[
    \begin{array}{c}
      \diag(\v{\xi}) + \diag(\v{\xi}) \v{M}^s \v{u}^s_{\v{\xi}} + \diag(\v{\xi}) \v{M}^d \v{u}^d_{\v{\xi}} \\
      \diag(\v{1}_n - \v{\xi}) + \diag(\v{1}_n - \v{\xi}) \v{M}^s \v{u}^s_{\v{\xi}} + \diag(\v{1}_n - \v{\xi}) \v{M}^d \v{u}^d_{\v{\xi}}
    \end{array}
    \right]
\end{align}
and thus by monotone convergence
$\v{u}^s_{\v{\xi'}} \geq \v{u}^s_{\v{\xi}}$ and
$\v{u}^d_{\v{\xi'}} \leq \v{u}^d_{\v{\xi}}$
which is the desired result.
\end{proof}

\section{Relation with threat index}
\label{app:Demange}

Here, we consider debt cross-holdings only, i.e. $\v{M}^s =
\v{0}$. Then, the model in \eq{XOS_vec} simplifies to
\begin{align}
  \v{s} &= \max\left\{\v{0}, \v{a} + \v{M}^d \v{r} - \v{d} \right\}, \\
  \v{r} &= \min\left\{\v{d}, \v{a} + \v{M}^d \v{r} \right\}
\end{align}
and as the right-hand side does not depend on $\v{s}$ it suffices to
consider the recovery values of debt $\v{r}$. Table \ref{tab:Demange}
provides the translation between our notation and the terminology of
\cite{Demange2018}.
\begin{table}[h]
  \centering
  \begin{tabular}{l|c|c|l}
    \multicolumn{2}{c|}{Demange model} & \multicolumn{2}{c}{Our model} \\ \hline
    Interpretation & Notation & Notation & Interpretation \\ \hline
    Total liabilities & $l^*_i = \sum_{j} l_{ij}$ & $d_i$ & Nominal debt \\
    Clearing ratio & $\theta_i \in [0, 1]$ & $\frac{r_i}{d_i}$ & NA \\
    Repayment & $\theta_i l^*_i$ & $r_i$ & Recovery value of debt \\
    Operating cash flow & $z_i$ & $a_i$ & External asset value \\
    Total cash fow & $a_i(\v{\theta}) = z_i + \sum_j \theta_j l_{ji}$ & $v_i = a_i + \sum_j M^d_{ij} r_j$ & Firm value \\
    Liabilities share & $\Pi_{ij} = \frac{l_{ij}}{l^*_i}$ & $M^d_{ji}$ & Investment fraction \\
    Default set & $i \in D$ & $\xi_i = 0$ & Solvency vector
  \end{tabular}
  \caption{Translation between our notation and the terminology of \cite{Demange2018}.}
  \label{tab:Demange}
\end{table}

Definition 1 by \cite{Demange2018} now states that first
$a_i(\v{\theta}) \geq z_i$ and second $\theta_i = 1$ or
$a_i(\v{\theta}) = \theta_i l^*_i$ hold. Translating into our
notation this means that $v_i \geq a_i$ and that the recovery value of
debt $r_i = \theta_i l^*_i$ either equals $d_i$ or $v_i$ if the firm
is insolvent. Thus, the constraints on the repayment ratio translate into
\begin{align}
  r_i &= \min\{d_i, v_i\} = \min\{d_i, a_i + \sum_j M_{ij}^d r_j\}
\end{align}
as in our model. Matching the definitions of $a_i(\v{\theta})$ and
$v_i$ we identify $\v{\Pi}$ with $(\v{M}^d)^T$.

Demange now considers the aggregate value of repayments $V$, i.e.
$V = \sum_i \theta_i l^*_i$ which we identify with $V = \sum_i
r_i$.
Thus, we can express $V$ as
$\v{0}^T \v{s} + \v{1}^T \v{r} = (\v{0}; \v{1})^T \v{x}$ where $\v{0}$
and $\v{1}$ are $n$-dimensional vectors of zeros and ones
respectively. Using corollary \ref{coro:fix_diff}, we can compute the
partial derivative of $V$ with respect to external asset values as
\begin{align}
  \frac{\partial V}{\partial \v{a}}
  &= \frac{\partial}{\partial \v{a}} (\v{0}; \v{1})^T \v{x}(\v{a}) \\
  & = (\v{0}; \v{1})^T \left[ \v{I}_{2 n \times 2 n} - \frac{\partial}{\partial \v{x}} \v{g}(\v{a}, \v{x}) \right]^{-1}
      \left[ \begin{array}{c} \diag(\v{\xi}) \\ \\ \diag(\v{1}_n - \v{\xi}) \end{array} \right] \\
  &= (\v{0}; \v{1})^T \left( \v{I}_{2 n \times 2 n} - \diag(\v{\xi}; \v{1}_n - \v{\xi})
    \left[
    \begin{array}{ccc}
      \v{0} & & \v{M}^d \\
      & & \\
      \v{0} &  & \v{M}^d
    \end{array}
                               \right] \right)^{-1}
      \left[ \begin{array}{c} \diag(\v{\xi}) \\ \\ \diag(\v{1}_n - \v{\xi}) \end{array} \right] \\
  &= \left(\v{0}^T, \v{1}^T (\v{I}_{n \times n} - \diag(\v{1}_n - \v{\xi}) \v{M}^d)^{-1} \right) 
    \left[ \begin{array}{c} \diag(\v{\xi}) \\ \\ \diag(\v{1}_n - \v{\xi}) \end{array} \right] \\
  &= \v{1}^T (\v{I}_{n \times n} - \diag(\v{1}_n - \v{\xi}) \v{M}^d)^{-1} \diag(\v{1}_n - \v{\xi}) \, .
\end{align}
Thus, denoting this row vector of partial derivatives with $\v{\mu}^T$
we find that
\begin{align}
  \v{\mu} &= \diag(\v{1}_n - \v{\xi})^T (\v{I}^T_{n \times n} - (\v{M}^d)^T \diag(\v{1}_n - \v{\xi})^T)^{-1} \v{1} \\
  &= \diag(\v{1}_n - \v{\xi}) (\v{I}^T_{n \times n} - \v{\Pi} \diag(\v{1}_n - \v{\xi}))^{-1} \v{1}
\end{align}
matches the definition proposed by \cite{Demange2018}.

\section{Solution for symmetric example}

Here, we compute the market values of debt and equity in the symmetric
example assuming that the single external asset $A_t$ follows a
geometric Brownian motion. In this case,
\begin{align}
  A_t &= a_0 e^{(r - \frac{1}{2} \sigma^2) t + \sigma \sqrt{t} Z}
\end{align}
where $Z \sim \mathcal{N}(0, 1)$, i.e. $A_t$ has a log normal
distribution with parameters
$\mu_A = \ln a_0 + (r - \frac{1}{2} \sigma^2) t$ and
$\sigma_A^2 = \sigma^2 t$.

\subsection{Equity and debt value}
\label{app:sol_one}

By \eq{val_Q} the market value of equity is then
\begin{align}
  s_t &= \E^Q_t[e^{-r \tau} s^*(A_T)] \\
  &= e^{- r \tau} \E^Q_t\left[ \E^Q_t[ s(A_T) | \xi ] \right] \\
  &= e^{- r \tau} \left( 0 \mathbb{P}^Q_t[\xi = 0] + \E^Q_t\left[ \frac{A^T - (1 - w^d) d}{1 - w^s} | \xi_T = 1 \right] \mathbb{P}^Q_t[\xi_T = 1] \right) \\
  &= e^{- r \tau} \left( \frac{1}{1 - w^s} \E^Q_t[ A_T | \xi_T = 1 ] - \frac{1 - w^d}{1 - w^s} d \right) \mathbb{P}^Q_t[\xi_T = 1] \, .
\end{align}
Defining $d_{\pm}$ as in \eq{sol_d_pm} the risk-neutral solvency
probability can be written as
\begin{align}
  \mathbb{P}^Q_t[\xi_T = 1] &= \mathbb{P}^Q_t[A_T \geq (1 - w^d) d] \\
  &= 1 - \Phi(- d_{-}) = \Phi(d_{-}) \, .
\end{align}
Similarly, the conditional expectation derived from the log normal
distribution for $A_T$ reads as
\begin{align}
  \E^Q_t[ A_T | \xi_T = 1 ] &= e^{\ln a_t + r \tau} \frac{\Phi(d_{+})}
                              { \mathbb{P}^Q_t[\xi_T = 1] }
\end{align}
Similarly the market value of debt is computed as
\begin{align}
  r_t &= e^{- r \tau} \left( \frac{1}{1 - w^d} \E^Q_t[A_T | \xi_T = 0] (1 - \mathbb{P}^Q_t[\xi_T = 1]) + d \; \mathbb{P}^Q_t[\xi_T = 1] \right)
\end{align}
with the conditional expectation
\begin{align}
  \E^Q_t[A_T | \xi_T = 0] &= e^{\ln a_t + r \tau} \frac{\Phi(- d_{+})}{ 1 - \mathbb{P}^Q_t[\xi_T = 1] } \, .
\end{align}
Finally, combining all equations we obtain the solution as given in
\eq{sol_one}.

\subsection{Greeks}
\label{app:sol_Greeks}

To compute the Greeks in this example, we first compute the partial
derivatives of the solution $\v{x}^* = (s^*, r^*)^T$ by corollary
\ref{coro:fix_diff}. Further, since all firms are symmetric and thus
either all solvent or all insolvent, from equations
\noeq{diff_solvent} and \noeq{diff_insolvent} we obtain:
\begin{align}
  \frac{\partial s^*}{\partial a}
  &= \left\{ \begin{array}{cl} 0 & \mbox{ if } \xi = 0 \\
               \frac{1}{1 - w^s} & \mbox{ if } \xi = 1 \end{array} \right. \\
  \frac{\partial r^*}{\partial a}
  &= \left\{ \begin{array}{cl} \frac{1}{1 - w^d} & \mbox{ if } \xi = 0 \\
               0 & \mbox{ if } \xi = 1 \end{array} \right. \, .
\end{align}
Summing together equity and debt and taking expectations, the systemic
risk index $\v{\pi}$ is given as
\begin{align}
  \v{\pi} &= E_t^Q[ \frac{\partial s^*}{\partial A_T} + \frac{\partial r^*}{\partial A_T} ] - 1 \\
          &= \frac{1}{1 - w^s} \mathbb{P}^Q_t[\xi = 1] + \frac{1}{1 - w^d} (1 - \mathbb{P}^Q_t[\xi = 1]) - 1 \, .
\end{align}
Next, we need the partial derivatives of $A_{T}$ with respect to
the parameters $\v{\theta} = (a_t, \sigma, r, \tau)$ of interest:
\begin{align}
  \frac{A_{T}}{\partial a_t}
  &= e^{(r - \frac{1}{2} \sigma^2) \tau + \sigma \sqrt{\tau} Z} = \frac{A_{T}}{a_t} \\
  \frac{A_{T}}{\partial \sigma}
  &= A_{T} \left( - \sigma \tau + \sqrt{\tau} Z \right) \\
  \frac{A_{T}}{\partial r}
  &= A_{T} \tau \\
  \frac{A_{T}}{\partial \tau}
  &= A_{T} \left( r - \frac{1}{2} \sigma^2 + \frac{1}{2} \frac{\sigma}{\sqrt{\tau}} Z \right) \, .
\end{align}
From \eq{greek_net} we finally obtain the Greeks:
\begin{align}
  \frac{\partial}{\partial a_t} \left( \begin{array}{c} s_t \\ r_t \end{array} \right)
  &= \frac{e^{- r \tau}}{a_t} \left( \begin{array}{c} \frac{1}{1 - w^s} \mathbb{P}^Q_t[\xi = 1] \E^Q_t[ A_T | \xi_T = 1 ] \\
                 \frac{1}{1 - w^d} \mathbb{P}^Q_t[\xi = 0] \E^Q_t[ A_T | \xi_T = 0 ] \end{array} \right) \\
  \frac{\partial}{\partial \sigma} \left( \begin{array}{c} s_t \\ r_t \end{array} \right)
  &= e^{- r \tau} \left( \begin{array}{c} \frac{1}{1 - w^s} \mathbb{P}^Q_t[\xi = 1] ( - \sigma \tau \E^Q_t[ A_T | \xi_T = 1 ] + \sqrt{\tau} \E^Q_t[ A_T Z | \xi_T = 1 ]) \\
                           \frac{1}{1 - w^d} \mathbb{P}^Q_t[\xi = 0] ( - \sigma \tau \E^Q_t[ A_T | \xi_T = 0 ] + \sqrt{\tau} \E^Q_t[ A_T Z | \xi_T = 0 ]) \end{array} \right) \\
  \frac{\partial}{\partial r} \left( \begin{array}{c} s_t \\ r_t \end{array} \right)
  &= - \tau e^{- r \tau} \left( \begin{array}{c} \left( \frac{1}{1 - w^s} \E^Q_t[ A_T | \xi_T = 1 ] - \frac{1 - w^d}{1 - w^s} d \right) \mathbb{P}^Q_t[\xi = 1] \\
  \frac{1}{1 - w^d} \E^Q_t[A_T | \xi_T = 0] \mathbb{P}^Q_t[\xi_T = 0] + d \; \mathbb{P}^Q_t[\xi_T = 1] \end{array} \right) \\
  &\phantom{=} + \tau e^{- r \tau} \left( \begin{array}{c} \frac{1}{1 - w^s} \mathbb{P}^Q_t[\xi = 1] \E^Q_t[ A_T | \xi_T = 1 ] \\
                                            \frac{1}{1 - w^d} \mathbb{P}^Q_t[\xi = 0] \E^Q_t[ A_T | \xi_T = 0 ] \end{array} \right) \\
  &= \tau e^{- r \tau} d \; \mathbb{P}^Q_t[\xi_T = 1]
    \left( \begin{array}{c} \frac{1 - w^d}{1 - w^s} \\ - 1 \end{array} \right) \\
  \frac{\partial}{\partial \tau} \left( \begin{array}{c} s_t \\ r_t \end{array} \right)
  &= - r e^{- r \tau} \left( \begin{array}{c} \left( \frac{1}{1 - w^s} \E^Q_t[ A_T | \xi_T = 1 ] - \frac{1 - w^d}{1 - w^s} d \right) \mathbb{P}^Q_t[\xi = 1] \\
  \frac{1}{1 - w^d} \E^Q_t[A_T | \xi_T = 0] \mathbb{P}^Q_t[\xi_T = 0] + d \; \mathbb{P}^Q_t[\xi_T = 1] \end{array} \right) \\
  &\phantom{=} + e^{- r \tau} \left( \begin{array}{c} \frac{1}{1 - w^s} \mathbb{P}^Q_t[\xi = 1] ( (r - \frac{1}{2} \sigma^2) \E^Q_t[ A_T | \xi_T = 1 ] - \frac{\sigma}{\sqrt{\tau}} \E^Q_t[ A_T Z | \xi_T = 1 ]) \\
                           \frac{1}{1 - w^d} \mathbb{P}^Q_t[\xi = 0] ( (r - \frac{1}{2} \sigma^2) \E^Q_t[ A_T | \xi_T = 0 ] - \frac{\sigma}{\sqrt{\tau}} \E^Q_t[ A_T Z | \xi_T = 0 ]) \end{array} \right)
\end{align}
All terms in these formulas are analytic except for the conditional
expectations $\E^Q_t[ A_T Z | \xi_T ]$ involving $A_T$ and $Z$
jointly, which we leave as an exercise to the reader. Instead, as a
sanity check, we compare the two Greeks, namely $\Delta$ and $\rho$
not involving this term with the corresponding results derived from
Black-Scholes formula.
\begin{align}
  \frac{\partial}{\partial a_t} \left( \begin{array}{c} s_t \\ r_t \end{array} \right)
  &= \left( \begin{array}{c} \frac{1}{1 - w^s} \Phi(d_{+}) \\
              \frac{1}{1 - w^d} \Phi(- d_{+}) \end{array} \right) \\
  &= \left( \begin{array}{c} \frac{1}{1 - w^s} \frac{\partial}{\partial a_t} C_{BS}(a_t, (1 - w^d) d, r, \tau, \sigma) \\
              \frac{1}{1 - w^d} \frac{\partial}{\partial a_t} \left( e^{-r \tau} (1 - w^d) d - P_{BS}(a_t, (1 - w^d) d, r, \tau, \sigma) \right) \end{array} \right) \\
  \frac{\partial}{\partial r} \left( \begin{array}{c} s_t \\ r_t \end{array} \right)
  &= \tau e^{- r \tau} (1 - w^d) d \; \Phi(d_{-})
    \left( \begin{array}{c} \frac{1}{1 - w^s} \\ - \frac{1}{1 - w^d} \end{array} \right) \\
  &= \left( \begin{array}{c} \frac{1}{1 - w^s} \frac{\partial}{\partial r} C_{BS}(a_t, (1 - w^d) d, r, \tau, \sigma) \\ \frac{1}{1 - w^d} \frac{\partial}{\partial r} \left( e^{-r \tau} (1 - w^d) d - P_{BS}(a_t, (1 - w^d) d, r, \tau, \sigma) \right) \end{array} \right)
\end{align}

\bibliographystyle{ormsv080}
\bibliography{arxiv} 

\end{document}